\documentclass[11pt]{article}
\usepackage[margin=1in]{geometry} 
\usepackage{amsmath,amsthm,amssymb,amsfonts}
\usepackage{caption}
\usepackage{hyperref}
\hypersetup{
    colorlinks,
    citecolor=blue,
    filecolor=black,
    linkcolor=red,
    urlcolor=black
}

\usepackage[usenames, dvipsnames]{color}

\usepackage{cleveref}

\newcommand{\R}{\mathbb{R}}

\newcommand{\N}{\mathbb{N}}

\newcommand{\E}{\mathbb{E}}
\newcommand{\F}{\mathbb{F}}

\newcommand{\per}[1]{\left(#1\right)}
\newcommand{\abs}[1]{\left|#1\right|}
\newcommand{\set}[1]{\left\{#1\right\}}
\newcommand{\prob}[2]{\mathrm{Pr}_{#1}\left[#2\right]}

\newcommand{\comment}[1]{}

\newtheorem{thm}{Theorem}[section]
\newtheorem*{thm*}{Theorem}
\newtheorem{lem}[thm]{Lemma}
\newtheorem{prop}[thm]{Proposition}
\newtheorem{defn}[thm]{Definition}
\newtheorem{cor}[thm]{Corollary}
\newtheorem{remark}[thm]{Remark}

\newtheorem*{lem*}{Lemma}

\newtheorem{innercustomthm}{Theorem}
\newenvironment{customthm}[1]
  {\renewcommand\theinnercustomthm{#1}\innercustomthm}
  {\endinnercustomthm}

\newtheorem{innercustomlem}{Lemma}
\newenvironment{customlem}[1]
  {\renewcommand\theinnercustomlem{#1}\innercustomlem}
  {\endinnercustomlem}

\newcommand{\poly}{\mathrm{poly}}

\newcommand{\dist}[2]{\text{dist}\per{{#1},{#2}}}

\newcommand{\weight}[1]{\mathrm{wt}(#1)}

\newcommand{\derivative}[2]{\Delta_{#1}{#2}}

\newcommand{\weightdistribution}[3]{W_{{#1},{#2}}\left(#3\right)}

\DeclareMathOperator{\bis}{bias}
\newcommand{\bias}[1]{\bis(#1)}

\newcommand{\reedmuller}[2]{\mathrm{RM}({#1},{#2})}

\newcommand{\polynomials}[2]{\reedmuller{m}{r}}

\newcommand{\binaryentropy}[1]{h\per{{#1}}}

\newcommand{\absoluteweight}[1]{{2^m\weight{#1}}}

\newcommand{\support}[1]{\text{supp}\per{#1}}

\DeclareMathOperator{\bscmath}{BSC}
\newcommand{\bsc}[1]{\bscmath_{#1}}

\DeclareMathOperator{\becmath}{BEC}
\newcommand{\bec}[1]{\becmath_{#1}}

\newcommand*\samethanks[1][\value{footnote}]{\footnotemark[#1]}

\title{On the Performance of Reed-Muller Codes with respect to Random Errors and Erasures}
\author{Ori Sberlo\thanks{Department of Computer Science, Tel Aviv University, Tel Aviv, Israel. The research leading to these results has received funding from the Israel Science Foundation (grant number 552/16) and from the Len Blavatnik and the Blavatnik Family foundation. Part of this work was done while the second author was a visitor at NYU.
} \and Amir Shpilka\samethanks[1]}
\date{}
\begin{document}
\maketitle

\begin{abstract}
This work proves new results on the ability of  binary Reed-Muller codes to decode from random errors and erasures. We obtain these results by proving improved bounds on the weight distribution of Reed-Muller codes of high degrees.


Specifically, given weight $\beta \in (0,1)$ we prove an upper bound on the number of codewords of relative weight at most $\beta$. We obtain new results in two different settings: for weights $\beta<1/2$  and for weights that are close to $1/2$. Our results for weights close to $1/2$ also answer an open problem posed by Beame et al.  \cite{beame2018bias}.

Our new bounds on the weight distribution imply that RM codes with $m$ variables and degree $\gamma m$, for some explicit constant $\gamma$, achieve capacity for random erasures (i.e. for the binary erasure channel) and for random errors (for the binary symmetric channel). Earlier, it was known that RM codes achieve capacity for the binary symmetric channel for degrees $r=o(m)$. For the binary erasure channel it was known that RM codes achieve capacity for degree $o(m)$ or $r\in[m/2 \pm O(\sqrt{m})]$. Thus, our result provide a new range of parameters for which RM achieve capacity for these two well studied channels. 

In addition, our results imply that for every $\epsilon>0$ (in fact we can get up to $\epsilon = \Omega\per{\frac{\sqrt{\log m}}{\sqrt m}}$) RM codes of degree $r < (1/2-\epsilon)m$  can correct a fraction of $1-o(1)$ random erasures with high probability. We also show that, information theoretically, such codes  can handle a fraction of $\frac{1}{2}-o(1)$ random errors with high probability. Thus, for example, given noisy evaluations of a degree $0.499m$ polynomial, it is possible to interpolate it even if a random $0.499$ fraction  of the evaluations were corrupted, with high probability. While the $o(1)$ terms are not the correct ones to ensure capacity, these results show that RM codes of such degrees are in some sense close to achieving capacity.

\end{abstract}
\thispagestyle{empty}
\newpage
\tableofcontents
\thispagestyle{empty}
\newpage
\pagenumbering{arabic}
\section{Introduction}
\subsection{Overview}
Reed-Muller (RM) codes were introduced by Muller \cite{muller1954application} and rediscovered shortly after by Reed \cite{reed1953class}, who also gave a decoding algorithm for them, and with time became one of the most well studied family of algebraic error correcting codes. 
Roughly, codewords of the RM code $\reedmuller{m}{r}_\F$ correspond to evaluation vectors of polynomials in $m$ variables of degree $r$ over a finite field $\F$. That is, a message is interpreted as the coefficient vector of an $m$-variate polynomial $f$, of degree at most $r$, over $\F$, and its encoding is the vector of evaluations $(f(a))_{a\in\F^m}$. The well known Hadamard code is simply $\reedmuller{m}{1}_{\F_2}$ and the family of Reed-Solomon codes corresponds to $\reedmuller{1}{k}_{\F_n}$ (for $k<n$).\footnote{$\F_q$ denotes the field with $q$ elements} In this work we only consider the case $\F=\F_2$\footnote{We only consider RM codes over $\F_2$ as this is the most difficult and interesting case for the questions we study.} and so we drop the subscript $\F$ and simply denote the $m$-variate degree-$r$ code as $\reedmuller{m}{r}$. RM codes have been extensively studied both in coding theory and in theoretical computer science, yet some of their basic and important properties are still unknown. One such important property is their \emph{weight-distribution}. Another fundamental question for which the answer is unknown is how well can RM codes handle random erasures or random errors. In this work we make progress on those two important questions showing that for a wide range of parameters RM codes are nearly optimal.

There are many motivating reasons to study RM codes. They are (arguably) the most natural family of linear error correcting codes, and indeed, they have been under investigation for over half a century. In addition, RM codes play a major role in a flora of applications in theoretical computer science. For example, 
in cryptography RM codes were used for constructing secret sharing schemes \cite{shamir1979share}, instance hiding schemes and private information retrieval protocols \cite{chor1995private,beaver1990hiding,gasarch2004survey}. In the theory of pseudorandomness they were used for constructing pseudorandom generators and randomness extractors \cite{bogdanov2010pseudorandom}. Similarly, they have found applications  in hardness amplification, in probabilistic proof systems and in many more areas. See e.g. \cite{abbe2015reed} for more applications.

Before discussing the coding theoretic questions we study here, we shall need some basic terminology concerning linear error correcting codes. A linear code over $\F_2$ can be viewed as a linear mapping $C:\F_2^k\to\F_2^n$ that maps messages of length $k$ to codewords of length $n$. It is convenient to abuse notation and identify the encoding map $C$ with its image $C(\F_2^k)$, which is a $k$-dimensional subspace. The rate of $C$ is the ratio $R(C)=k/n$, which, intuitively, captures the average amount of information each bit of the codeword contains. Alternatively, one can think of the rate of the code as a measure of the redundancy in the encoding - the smaller the rate is the more redundant the code is. As multilinear monomials form  a basis to the space of multilinear  polynomials, the rate of $\reedmuller{m}{r}$ is $\binom{m}{\leq r}/2^m$ where $\binom{m}{\leq r} \triangleq \sum_{i=0}^{r} \binom{m}{i}$.

Another important property of a code $C$ is its \emph{weight distribution}. Given a codeword $w \in C$ its (normalized) weight is the fraction of its nonzero coordinates. E.g., in the case of RM codes, the weight of a codeword $f \in \reedmuller{m}{r}$ equals the (relative) number of its nonzero evaluations: $\weight{f}=\Pr_{x \in \F_2^m}[ f(x) \neq 0]$. The weight distribution of a code is the function counting the number of codewords of any given weight in the code (see \Cref{def:wt}). 
Besides being a natural property to study, the weight distribution of a code plays an important role when analyzing its resilience to errors.

A major problem of coding theory is to construct efficient and \emph{optimal} codes that can handle as many errors as possible. I.e., that given a corrupted encoding of a message there is an algorithm for recovering the  message. To understand what optimal means we heed to discuss the model of corruptions. In this work we study the two most well known models of corruption - that of erasing a symbol and that of flipping a symbol. But another important characteristic of the corruption model is whether the errors are random or worst case. These two models were introduced in the seminal works of Shannon \cite{shannon2001mathematical} and Hamming   \cite{hamming1950error}.

In the worst-case (or adversarial) setting, introduced by Hamming \cite{hamming1950error}, we allow an adversary to corrupt a fraction $\delta$ of the coordinates. Decoding in this setting is possible if the  minimal (normalized) Hamming distance between any two codewords is larger than $2\delta$. This quantity is also known as the (relative) minimum distance of the code. Assuming the code is linear, it is not hard to see that its minimum distance equals the minimal weight of a nonzero codeword. Hence, the performance of a linear code in Hamming's model is completely determined by its weight distribution. It is not hard to see that the (relative) minimum distance of $\reedmuller{m}{r}$  is $2^{-r}$. Thus, $\reedmuller{m}{r}$ can tolerate only a small amount of adversarial errors for large $r$.

The model of random corruptions, which is the one we focus on in this paper, was introduced by Shannon in his influential work \cite{shannon2001mathematical}. In this setting we assume that each coordinate is randomly and independently mapped to a symbol from a fixed alphabet (not necessarily binary) according to some fixed probability distribution. Every such probability distribution gives rise to a random corruption model, which is called a \emph{channel}. The simplest and most classical examples of channels for binary codes are the Binary Erasure Channel (BEC) and the Binary Symmetric Channel (BSC). In the $\bec{p}$, each coordinate is replaced with a question mark `?' with probability $p$. This can be thought of as erasing the coordinate. In the $\bsc{p}$, each coordinate is flipped with probability $p$. We sometimes abbreviate and just say random errors instead of $\bsc{p}$ or random erasures instead of $\bec{p}$, when $p$ is either clear from the context or immaterial. Note that the crucial difference between the BEC and the BSC is that in the BEC we know where the errors are (these are all coordinates with a question mark) whereas in the BSC model we do not know which coordinates were corrupted. Unlike the worst case model, here we can only require decoding with high probability as it may be the case that with some tiny probability the entire codeword is erased (i.e. all coordinates are replaced with question marks) or that the errors are such that they take us from one codeword to another. In his original work Shannon also asked what is the maximal rate of a code that (with high probability) can recover from random errors introduced by a given channel. This maximal rate is called the \emph{capacity} of the channel. Shannon proved that the capacity of the $\bsc{p}$ is $R=1-\binaryentropy{p}$, where $\binaryentropy{\cdot}$ is the binary entropy function and $p \leq 1/2$.\footnote{$\binaryentropy{p} = -p \log p - (1-p)\log(1-p)$, for $p \in (0,1)$, and $\binaryentropy{0}=\binaryentropy{1}=0$.} In other words, for every rate $R< 1-\binaryentropy{p}$ there is a code that can recover any message from random errors introduce by the channel, with high probability, and no code of rate $R> 1-\binaryentropy{p}$ can do so. For the $\bec{p}$ it was shown by Elias \cite{elias1955coding} that the capacity is $R=1-p$. Families of codes whose rate approach the capacity (in the limit as the block length goes to infinity) are called capacity achieving codes.  Unfortunately, Shannon only showed the existence of codes that achieve capacity without presenting an explicit construction. Thus, constructing capacity achieving codes that are easy to encode and decode has been a major problem in coding theory. A breakthrough was made by Arikan \cite{arikan2009channel} with the introduction of polar codes. Arikan's construction achieves capacity for a wide variety of channels including the BEC and BSC. Due to their similarity to RM codes, the introduction of polar codes  brought back to the spotlight the classical question of whether RM codes also achieve capacity for the BEC and the BSC. Indeed, despite of their poor performance in the adversarial error model, it is believed that RM codes achieve capacity for both the BEC and the BSC. This belief is also supported by empirical studies suggesting that RM codes perform even better than polar codes \cite{arikan2009performance}. 
In fact, for some setting of parameters it was recently proved that RM codes achieve capacity. Abbe et al. \cite{abbe2015reed} proved that RM codes achieve capacity for the BEC for rates going to $0$ or to $1$ and for the BSC for rates going to $0$. A beautiful work of Kumar et al. showed that RM codes with constant rate (the most interesting range of parameters in coding theory) achieve capacity for the BEC \cite{kumar2015reed}. However, both works leave open a wide range of parameters, especially for the BSC. See \autoref{table} on page \pageref{table} for a summary of known results. 

Another intriguing question is even if RM codes do not achieve capacity, what is the amount of random errors or erasure from which they can successfully decode? One of the difficulties in answering this question and in showing that a family of codes achieves capacity for the BEC or the BSC is that for these two important channels there is no one parameter that governs the ability of the code to recover from errors. That being said, it is clear that the weight distribution of a code is intimately related to recovering from errors, and this holds even in Shannon's model. In order to exemplify this statement consider the case of random erasures. It is not hard to see that a linear code can recover from an erasure pattern if and only if there does not exist a codeword supported on the erasure locations (for details see \Cref{recover from erasure equivalent to support of codewords}). Therefore, if a code has many codewords with small support, or equivalently low weight, then most likely it will not be able recover from  random erasures. This observation can be used to analyze the probability of recovering from random erasures and a similar analysis can be made in the case of random errors (e.g, See \cite{poltyrev1994bounds,abbe2015reed}). Thus, from this point of view, to understand whether RM codes achieve capacity for the BEC and the BSC it is important to understand their weight distribution.

Computing the weight distribution of RM codes is a well known problem that is open in most ranges of parameters. In 1970 Kasami and Tokura \cite{kasami1970weight} characterized all codewords of weight up to twice the minimum distance. This was  later improved in \cite{kasami1976weight} to all words of weight up to $2.5$ times the minimal distance. No progress was made for over thirty years until a breakthrough result of Kaufman, Lovett and Porat \cite{kaufman2012weight}  gave, for any constant degree $r=O(1)$,  asymptotically tight bounds on the weight distribution of RM codes of degree $r$. Unfortunately, as the degree gets larger, their estimate becomes less and less tight. Abbe, Shpilka and Wigderson \cite{abbe2015reed} managed to get better bounds for degrees up to $m/4$, which they used to show that RM codes achieve capacity for the BEC and the BSC for degrees $r=o(m)$. Recently Samorodnitsky proved new bounds on the weight distribution of codes whose duals are capacity achieving for the BEC  \cite{DBLP:journals/corr/abs-1809-09696}. When combined with the result of \cite{kumar2015reed} this implies bounds on the weight distribution of RM codes of constant rate (i.e. degrees $r\in [m/2 \pm O(\sqrt{m})]$).

While the results of \cite{kaufman2012weight,abbe2015reed,DBLP:journals/corr/abs-1809-09696} mostly give non trivial bounds for constant weights $\beta<1/2$,  it is an intriguing question to better understand the weight distribution for weights closer to $1/2$.  When studying weights close to $1/2$ it is more convenient to consider the bias of a polynomial, rather than its weight. The bias of a polynomial is the difference between the fraction of its zero evaluations to the fraction of its nonzero evaluations. Thus, having bias at most $\epsilon$ corresponds to having weight at least $\frac{1- \epsilon}{2}$. It is easy to see (due to symmetry) that  the expected bias of a random polynomial is zero, and, Intuitively, we expect a random polynomial to have bias close to zero, in the same way that a random function is nearly unbiased. It is therefore natural to ask how concentrated around zero is the bias of a random polynomial. To the best of our knowledge, besides what is implied by \cite{kaufman2012weight,abbe2015reed,DBLP:journals/corr/abs-1809-09696}, the other known result on the weight distribution in this regime, over $\F_2$, is due to Ben Eliezer, Hod and Lovett who gave an upper bound on the number of $m$-variate, degree $r$ polynomials of bias at least $2^{-c \cdot\frac{m}{r}}$, where $c$ is some positive constant depending on the ratio $r/m$ \cite{ben2012random} (this result was later extended to  other prime fields in  \cite{beame2018bias}). 
Thus, prior to this work no strong bounds were known for linear degrees and sub-constant bias.
Besides being a natural question, our proofs demonstrate that  improving the bound on the weight distribution in this regime leads to improved results on the performance of RM codes.

\subsection{Our results}
\subsubsection{Weight distribution}
We prove new results on the weight distribution of RM codes.
Specifically, we prove new upper bounds on the number of polynomials of weight at most $\beta$, for $\beta<1/2$, and on the number of polynomials that have bias at least $\epsilon$ (this result only holds for degrees $r < m/2$). We complement this result by proving a lower bound on the number of polynomials that have bias at least $\epsilon$. 

%
%

To state our results we shall need the following notation.
\begin{defn}\label{def:wt}
We denote $\weight{f} = \E_x[f(x)]= \mathrm{Pr}_x[f(x) =1]$ and $\bias{f} = \E_x[(-1)^{f(x)}]$.    
For any $\beta \in [0,1]$ we let,
$\weightdistribution{m}{r}{\beta} \triangleq \abs{\set{f \in \polynomials{m}{r}: \weight{f}\leq \beta}}$.
\end{defn}

Our first result is an upper bound on the weight distribution for weights smaller than $1/2$.

\begin{thm}
\label{main thm - low weight}
Let $r, m,\ell \in \N$ such that $r \leq m$ and write $\gamma = r/m$. Then,
$$W_{m,r} (2^{-\ell})\leq \exp_2\per{O(m^4) + 17(c_{\gamma}\ell+d_{\gamma})\gamma^{\ell-1}\binom{m}{\leq r}} \;,$$
where $c_{\gamma} = \frac{1}{1-\gamma}$ and $d_{\gamma} = \frac{2-\gamma}{(1-\gamma)^2}$.
\end{thm}

This bound  improves an earlier result by Abbe et. al. \cite{abbe2015reed} in two aspects. First, our result applies to any degree $r$, while their result only holds for $r<m/4$. Second, the leading term in the exponent in our result is $O\per{\ell \gamma^{\ell-1}\binom{m}{\leq r}}$ as opposed to $O\per{\ell^4 \gamma^{\ell-1}\binom{m}{\leq r}}$ in  \cite{abbe2015reed} (to see this compare \Cref{main thm - low weight} to Theorem 3.3 in \cite{abbe2015reed}).

Recently \cite{DBLP:journals/corr/abs-1809-09696} proved new results on the weight distribution of codes whose duals achieve capacity.\footnote{We shall use $\log x$ to denote the base 2 logarithm and $\ln x$ for the natural logarithm.}

\begin{thm}[Proposition 1.6 in \cite{DBLP:journals/corr/abs-1809-09696}]\label{thm:sam}
Let $C$ be the dual of a linear code $C^\perp$, of length $n$, achieving BEC capacity. Let $R = R (C)$ be the rate of $C$. Let $(b_0, ..., b_n)$ be the distance distribution of $C$.\footnote{i.e. $b_i$ is the number of codewords of weight $i/n$.} Then for all $0 \leq i \leq n$ it holds that
$$b_i \leq 2^{o(n)}\cdot \left(\frac{1}{1-R}\right)^{i\cdot 2\ln 2}\;.$$
\end{thm}
As $\reedmuller{m/2}{m}$ achieves capacity for the BEC and is (more or less) its own dual we get that
$$\weightdistribution{m}{m/2}{2^{-\ell}} \leq \exp_2\left(2^{-\ell+1} \ln 2 \cdot 2^m  \right) = \exp_2\left(2^{-\ell+2} \ln 2 \cdot {m \choose \leq m/2}  \right) \;.$$
This result is better than what \Cref{main thm - low weight} gives for $r=m/2$ (i.e. $\gamma=1/2$). Nevertheless we note that even if it was the case that RM codes achieve capacity for the BEC for every degree, then for degrees $r=\gamma m$, for $\gamma<1/2$, the bound in \Cref{main thm - low weight} will be better than the one in \Cref{thm:sam} as the leading term in the exponent in \Cref{main thm - low weight}  is $O(\ell \gamma^{\ell-1}\binom{m}{\leq r})$ whereas  \Cref{thm:sam} gives $O( 2^{-\ell}\binom{m}{\leq r})$ as leading term.

Next, we state our upper bounds on the number of polynomials of bias at least $\epsilon$. We first state our result for the case that $r <m/2$.
\begin{thm}
\label{main thm - low bias estimation}
Let $\ell,m,\in \N$ and let $0 < \gamma(m) < 1/2 - \Omega\per{\sqrt{\frac{\log m}{m}}}$ be a parameter (which may be constant or depend on $m$) such that $\frac{\ell+\log\frac{1}{1-2\gamma}}{(1-2\gamma)^2} = o(m)$. Then, 
$$\weightdistribution{m}{\gamma m}{\frac{1-2^{-\ell}}{2}} \leq \exp_2\per{O(m^4)+\per{1-2^{-c(\gamma,\ell)}}\binom{m}{\leq r}} \;,$$
where  $c(\gamma,\ell) = O\per{\frac{\gamma^2 \ell  +  \gamma \log(1/1-2\gamma)}{1-2\gamma} + \gamma}$.
\end{thm}
\begin{remark}
To make better sense of the parameters in the theorem we note the following.
\begin{itemize}
\item When $\gamma<1/2$ is a constant, $c(\gamma,\ell) =O(\ell)$.
\item The bound is meaningful up to degrees $\left(\frac{1}{2} - \Omega\left(\frac{\sqrt{\log m}}{\sqrt m}\right) \right)m$, but falls short of working for constant rate RM codes.

\item
For $\gamma$ which is a constant the upper bound is applicable to $\ell = o(m)$ (in fact it is possible to push it all the way to some $\ell = \Omega(m)$). For $\gamma$ approaching $1/2$, i.e $\gamma = 1/2 - o(1)$, there is a trade-off between how small the $o(1)$ is and the largest $\ell$ for which the bound is applicable to. Nevertheless, even if $\gamma = 1/2 - \Omega\per{\sqrt{\frac{\log m}{m}}}$ the lemma still holds for $\ell = \Omega(\log m)$ (i.e, for a polynomially small bias).
\end{itemize}
\end{remark}


To the best of our knowledge, besides the work of \cite{kaufman2012weight} that speaks of constant degrees and \cite{abbe2015reed,DBLP:journals/corr/abs-1809-09696}  that do not say much when the bias is smaller than $1/2$, the only other relevant result is the following bound of of Ben Eliezer, Hod and Lovett \cite{ben2012random}.
\begin{thm*}
[Lemma 2 in \cite{ben2012random}]
Let $m,r \in \N$ and $\epsilon > 0$ such that $r \leq (1-\delta)m$. Then there exist positive constants $c_1,c_2$ (which depends solely on $\delta$) such that,
$$\prob{f}{\abs{\bias{f}} \geq 2^{-c_1 \frac{m}{r}}} \leq \exp_2\per{-c_2 \binom{m}{\leq r}} \;,$$
where the probability is over a uniformly random polynomial with $m$ variables and degree $\leq r$.
\end{thm*}
We see that for linear degrees ($r=\Omega(m)$) this result  gives a  bound on the number of polynomials (or codewords) that have at least some constant bias, whereas \Cref{main thm - low bias estimation} holds for a wider range of parameters and in particular can handle bias which is nearly exponentially small. 
We now state our upper bound for arbitrary degrees.

\begin{thm}
\label{thm : weak concentration of bias for all gamma}
Let $r \leq m\in \N$ and $\epsilon > 0$. Then,
$$\prob{f \sim \reedmuller{m}{r}}{\abs{\bias{f}} > \epsilon} \leq 2\exp\per{-\frac{2^{r}\epsilon^2}{2}} \;.$$
\end{thm}

Compared to the result of  \cite{ben2012random} this gives a weaker estimate as the upper bound does not show that the number of codewords is at most the size of the code to some constant power smaller than $1$. On the other hand our result holds for sub-constant bias as well, and in fact it gives meaningful bounds also for exponentially small bias. 

Finally we note that our results answer a question posed by Beame, Oveis Gharan and Yang  \cite{beame2018bias}. They asked whether it is possible to obtain similar bounds to those of \cite{ben2012random} where the bias does not have  $ \frac{m}{r}$ in the exponent. The results stated in  \Cref{main thm - low bias estimation} and \Cref{thm : weak concentration of bias for all gamma} provide  such bounds.\\

We next state our lower bound on the number of polynomials of bias at least $\epsilon$. 
\begin{thm}
\label{main thm: lower bound for bias}
Let $20\leq r \leq m,\in \N$. Then for any integer $\ell <r/3$ and sufficiently large $m$ it holds that
$$\abs{f \in \reedmuller{m}{r} : \bias{f} \geq 2^{-\ell}} \geq \frac{1}{2}\cdot \exp_2\per{\sum_{j=1}^{\ell-1}\binom{m-j}{\leq r-1}} \;.$$
\end{thm}
Comparing the upper bound in \Cref{main thm - low bias estimation} to \Cref{main thm: lower bound for bias} we see that there is a gap between the two bounds. Roughly, the lower bound on number of polynomials that have bias at least $\epsilon$ matches the upper bound corresponding to bias at least $\sqrt{\epsilon}$. This may be a bit difficult to see when looking at \Cref{main thm - low bias estimation} but see \Cref{rem:comparison} for a qualitative comparison.

\subsubsection{Capacity results for Reed-Muller codes}
\label{section : known results rm capacity}
There are three settings of parameters for which RM codes were known to achieve capacity\footnote{We formally define the notion of achieving capacity in \Cref{section : preliminaries}.} in the BEC: Degrees $r(m) = o(m)$ (See Theorem 1.2 in \cite{abbe2015reed}); Constant rate, i.e., when
the degree is  $r(m)=\frac{m}{2} \pm O(\sqrt{m})$ (See \cite{kumar2015reed}); Degrees $r(m) =m - o(\sqrt{m /\log m})$ (See Theorem 1.4 in \cite{abbe2015reed}).
%
%
Perhaps surprisingly, these results are obtained via very different approaches. As for errors, the situation is even worse and prior to this work the only setting for which it was known that RM codes achieves capacity was the low degree setting, $r(m) = o(m)$ (Theorem 1.7 of \cite{abbe2015reed}). 

Using our new upper bounds on the weight distribution of RM codes we obtain the following improvements on the low degree regime.
\begin{thm}
\label{main thm - capacity for bec}
For any $\gamma \leq 1/50$ the RM code $\reedmuller{m}{\gamma m}$ achieves capacity for the BEC.
\end{thm}

\begin{thm}
\label{main thm - capacity for bsc}
For any $\gamma \leq 1/70$ the RM code $\reedmuller{m}{\gamma m}$ achieves capacity for the BSC.
\end{thm}

The next table summarizes the range of parameters for which RM codes achieve capacity. Our results for each channel appear in the right most column.
\begin{center}
  \begin{tabular}{ |c |  c | c  | c|c|}\hline 
    &&&& \\
    & $r=o(m)$ & $r\in [m/2 \pm O(\sqrt{m})]$ & $r = m- o(\sqrt{m/\log m})$ & This Work\\
    &&&& \\

   \hline
  &&&& \\
BEC  & \cite{abbe2015reed}  & \cite{kumar2015reed} & \cite{abbe2015reed} &$r\leq m/50$\\ 
&&&& \\ 
    \hline 
    &&&&\\
   BSC & \cite{abbe2015reed}   & ? & ? &$r\leq m/70$ \\ 
    &&&&\\
    \hline
  \end{tabular}
    \captionof{table}{Capacity results for RM codes}\label{table}
\end{center}

\subsection{Reed-Muller codes under random noise}

Finally, we show that although we do not know whether RM codes of higher degrees achieve capacity they can nevertheless handle a large fraction of random errors and erasures up to rates polylogarithmic in the length of the code. 

\begin{thm}
\label{main thm - noise for bec}
For any $\gamma < 1/2 - \Omega\per{\frac{\sqrt{\log m}}{\sqrt{m}}}$, $\reedmuller{m}{\gamma m}$ can efficiently decode a fraction of $1-o(1)$ random erasures.
\end{thm}

\begin{thm}
\label{main thm - noise for bsc}
For any $\gamma < 1/2 - \Omega\per{\frac{\sqrt{\log m}}{\sqrt{m}}}$ the  maximum likelihood decoder for $\reedmuller{m}{\gamma m}$ can decode from a fraction of $1/2-o(1)$ random errors.
\end{thm}

Observe  that the only difference between these results and what we would have achieved had we known that RM codes achieve capacity for such degrees, is the $o(1)$ term. The $o(1)$ term in \Cref{main thm - noise for bec} and \Cref{main thm - noise for bsc} is larger than the corresponding term in capacity achieving codes.
Similarly, the work of Saptharishi, Shpilka and Volk \cite{DBLP:journals/tit/SaptharishiSV17} (and the improved version in \cite{DBLP:conf/soda/KoppartyP18}) that gave an efficient decoding algorithm for $\reedmuller{m}{O(\sqrt m)}$ was able to decode from a fraction of $1/2-o(1)$ random errors and here too the $o(1)$ term is not the one guaranteed from the fact that these codes achieve capacity.

Finally, we note that a result similar to \Cref{main thm - noise for bec} could have been obtained by the authors of \cite{abbe2015reed} (although they did not study this problem), using their results on weight distribution, albeit for degrees up to $m/4$. The reason we are able to push this to all degrees up to (roughly) $m/2$ is that our bounds on the weight distribution hold for such degrees as well. For the BSC, we needed a stronger bound on the error of the maximum likelihood decoder than the one in \cite{abbe2015reed} (See \cref{equation for bsc}), in addition to our new results on the weight distribution.

\subsection{Proof strategy}

We first explain how we approach the problem of proving that RM codes achieve capacity for the BEC and the BSC. We basically follow the same general strategy that Abbe et. al. \cite{abbe2015reed} applied in their proof for the low degree regime (which is similar to the approach of \cite{poltyrev1994bounds}). For simplicity, let us focus on the  case of random erasures. In  \cite{abbe2015reed} the authors  used the well known fact that the following is an upper bound on the probability that $\reedmuller{m}{r}$ cannot recover from $s$ random erasures
\begin{equation}
\label{ASW fundamental quantity}
\sum_{\beta} (1-\beta)^{2^m -s} \cdot \weightdistribution{m}{r}{\beta} \;.
\end{equation}
Thus, proving that RM codes achieve capacity for random erasures reduces to showing that the above tends to zero as $m$ tends to infinity. Intuitively, we want to show that $\weightdistribution{m}{r}{\beta}$ decays faster than $(1-\beta)^{2^m -s}$ so that the sum remains very small. In order to estimate the sum, the authors of \cite{abbe2015reed} partition the summation over $\beta$ to the dyadic intervals $[2^{-k-1}, 2^{-k}]$ and show that each such interval sums to a small quantity. To show this they use the following elementary upper bound,
$$\sum_{2^{-k-1}  \leq \beta \leq 2^{-k}} (1-\beta)^{2^m-s} \cdot \weightdistribution{m}{r}{\beta} \leq (1-2^{-k-1})^{2^m-s} \weightdistribution{m}{r}{2^{-k}} \;.$$
Inspecting their argument, we learn that the problem in extending it to larger degrees lies in the regime of weights which are very close to $1/2$. Roughly, the leap from small bias (i.e, weight roughly $1/2$) to weight $1/4$ is too crude and loses too much information. To overcome this, we partition the interval $[1/4,1/2]$ further to smaller dyadic intervals. Specifically, we start with polynomials of bias $\delta$, for some subconstant $\delta$, and double the bias until we reach bias $1/2$ (equivalently, weight $1/4$). Finally, we use \Cref{main thm - low bias estimation} to estimate the sum over these intervals. For the case of BSC we  need to strengthen the upper bound on the decoding error of \cite{abbe2015reed}. This turns to be a rather delicate task and in particular we end up using the infinite Taylor expansion of the binary entropy function and we cannot just take the first few terms in it. See \Cref{equation for bsc}.

The proofs of Theorems~\ref{main thm - noise for bec} and \ref{main thm - noise for bsc} are similar in spirit. They are based on estimating the relevant sums. The main point being that by not insisting on the exact probability of errors or erasures coming from the capacity calculations, but rather altering it by subtracting from it a $o(1)$ term, is sufficient to show that the probability of error in the recovery algorithms tend to zero. 

We now explain how we get our improved bounds on the weight distribution.
Although not stated explicitly in this way before, the main idea in \cite{kaufman2012weight} and \cite{abbe2015reed}, and in this work as well, is that in order to bound the number of polynomials of certain weight we would like to find a relatively small $\delta$-net, in the space of all functions $\F_2^m \rightarrow \F_2$, with respect to the Hamming distance, such that all low weight polynomials are contained in small balls around the elements of the $\delta$-net. Similar to \cite{kaufman2012weight,abbe2015reed} the elements of our $\delta$-net are not going to be low degree polynomials themselves. To get effective bounds on the number of low weight/bias polynomials we would like the $\delta$-net to be as ``efficient'' as possible. This means that we would like the $\delta$-net to be relatively small, that no ball around an element of the $\delta$-net should contain too many polynomials, and that the union of the balls cover all low weight/bias polynomials. 
We shall therefore focus on the following two important parameters of the $\delta$-net: the size of the $\delta$-net and the maximum number of polynomials from our set that are contained in each ball. Intuitively, these two quantities are inversely correlated -- the larger the $\delta$-net is the better it approximates our set. 

Adopting this point of view, \cite{kaufman2012weight} constructed a $\delta$-net such that each ball contains at most one low weight polynomial. Their beautiful observation was that centers for the $\delta$-net, for degree-$r$ polynomials, can be described as explicit functions of some number of polynomials of lower degree. Counting the number of such possible representations they obtained a bound on the size of the $\delta$-net. However, since they insisted on having at most one low weight polynomial in every ball this resulted in a relatively large net. In \cite{abbe2015reed} the net was constructed in such a way that  balls contained many low weight polynomials (though not too many). This allowed Abbe et al. to have a smaller net and consequently they obtained a significant improvement over  \cite{kaufman2012weight}. The primary observation in \cite{abbe2015reed}  is that the maximum number of low weight polynomials  that are contained in each ball, is related to the weight distribution and so a recursive approach can be taken. That is, if we consider balls of radius $\delta$ then the number of polynomials in such a ball is at most the number of polynomials of weight at most $2\delta$. The centers of the balls were constructed in a similar fashion to \cite{kaufman2012weight}, taking into account that we allow balls to contains several polynomials, and  \cite{abbe2015reed} further improved over  \cite{kaufman2012weight} by making a tighter analysis of the number of possible centers.

In this work we take a similar approach but we improve upon \cite{abbe2015reed} in several ways. First, we  show that the counting argument of  \cite{abbe2015reed} is not tight for low weight polynomials and give a tighter analysis. We then focus on low bias polynomials.  \cite{kaufman2012weight,abbe2015reed} did not give good bounds in this regime and obtaining improved results for such polynomials was essential for proving the results on the capacity (this should be  clear from the explanation above). To analyze polynomials having low bias we give
yet a tighter analysis of the possible number of possible elements in the $\delta$-net. Very roughly, each element of the net can be represented as an explicit function of several derivatives of one of the low bias polynomials (this is true not just for low bias polynomials). One idea in the improvement of \cite{abbe2015reed} over \cite{kaufman2012weight} is that derivatives of polynomials can be represented as polynomials in fewer variables. This allowed some saving in the counting argument. We make further improvement by noting that different derivatives contain information about each other. This allows us to get a better control of the amount of information encoded in the list of derivatives and as a result obtain a better bound on the size of the net. Furthermore, in this regime of parameters, both   \cite{kaufman2012weight,abbe2015reed}  essentially picked a net that has at most one low bias polynomial in each ball. We show that by allowing several polynomials per ball we can get another significant improvement in the size of the net. 


\subsection{Organization}
We start by describing the model of random erasures and random errors (Shannon's model), and the notion of capacity achieving codes (see \Cref{section : preliminaries}). In \Cref{section : weight distribution} we prove our main results on weight distribution of RM codes, which are \Cref{main thm - low weight} , \Cref{main thm - low bias estimation} and \Cref{main thm: lower bound for bias}. In \Cref{section: rm achieve capacity bec} we prove that for any $\gamma \leq 1/50$ the RM code $\reedmuller{m}{\gamma m}$ achieves capacity for the BEC and in \Cref{sec:RM-erasures-1/2} we prove that RM codes of degree $(1/2-o(1))m$ can recover from $1-o(1)$ random erasures.

In \Cref{sec:BSC} we prove \Cref{equation for bsc}. In \Cref{section: RM achieve capacity for bsc}
 we prove that for any $\gamma \leq 1/70$ the RM code $\reedmuller{m}{\gamma m}$ achieve capacity for the BSC and in \Cref{sec:RM-errors-1/2} we show that RM codes of degrees $(1/2-o(1))m$ can recover from $1/2-o(1)$ random erasures.

In \Cref{sec:discuss} we discuss the result and some open problems.

\section{Preliminaries}
\label{section : preliminaries}
\subsection{Reed-Muller codes}\label{sec:RM-def}
Recall that a linear binary code of block length $n$ and dimension $k$ is a linear subspace $C$ of $\F_2^n$ with dimension $k$. 
\begin{defn}
The code $\reedmuller{m}{r} \subseteq \F_2^{2^m}$ is defined as all evaluation vectors of multi-variate polynomials over $\F_2$ with $m$ variables and degree at most $r$. That is, for every such polynomial $f$ there is a corresponding codeword $(f(a))_{a\in\F^m_2}$.
\end{defn}

It is clear that  $\reedmuller{m}{r}$ is a linear code with blocklength is $2^m$ and rate $R = 2^{-m}\binom{m}{\leq r}$.



When $r(m)$ is a fixed integer function, e.g $r(m) = 3$, $r(m) = \lceil\sqrt{m}\rceil$, $r(m) = \lceil m/10\rceil$, rather than looking at specific values of $r$, and $m$ we will be thinking about the family of codes $\reedmuller{m}{r}$ as $m$ goes to infinity.
In particular, we will be mostly interested in the setting $r(m) = \lceil\gamma m \rceil$ where $\gamma$ is either a positive constant or a fixed function of $m$. To keep the notation simple, we often write $r(m) = \gamma m$ instead of $r(m) = \lceil\gamma m \rceil$. Note that when $\gamma$ is a constant, the rate  tends to zero if $\gamma < 1/2$, equals $1/2$ if $\gamma = 1/2$, and tends to $1$ if $\gamma > 1/2$.

\subsection{Shannon's noise model}
We now formally define the $\bec{p}$ and the $\bsc{p}$, and formalize the task of recovering from random erasures and random errors.
\begin{defn}
Given $p \in [0,1]$ and $y \in \set{0,1}^n$ define the following two distributions $\bec{p}(y)$, $\bsc{p}(y)$ as follows: 
\begin{itemize}
    \item 
    $\bec{p}(y)$: Every bit of $y$ is replace with `?' with probability $p$ independently and remains unchanged otherwise.
    \item
    $\bsc{p}(y)$: Every bit of $y$ is flipped with probability $p$ independently and remains unchanged otherwise.
\end{itemize}
Note that $\bec{p}(y)$ is a distribution over $\set{0,1,?}^n$ while $\bsc{p}(y)$ is over $\set{0,1}^n$.
\end{defn}

We now explain the  tasks of recovering from random errors or erasures. We start with erasures: given $z \sim \bec{p}(y)$ recover $y$. A necessary condition for this to be possible is that there exists a unique codeword $y$ that agrees with $z$ on its non-erased coordinates, i.e either $z_i = y_i$ or $z_i = ?$. In fact, for linear codes (in particular RM code), assuming there exists a unique codeword $y$ that agrees with $z$ on its non-erased coordinates, recovering $y$ from $z$ can be done efficiently by solving a system of linear equations. 
Recovering from errors is more subtle. Suppose $z \sim \bsc{p}(y)$ then $y$ is possibly any string (though some strings are more likely than others). Therefore, to formally define recovering from random errors we need to specify the decoding procedure. One possible choice is the minimum distance decoder which decodes $y$ to the closest codeword,
$$\mathrm{MD}(z) = \underset{y \in C}{\text{argmin}}\set{\dist{y}{z}} \;.$$
Another possible choice is the maximal-likelihood decoder in which we decode $y$ to the most probable codeword,
$$\mathrm{ML}(z) = \underset{y \in C}{\text{argmax}}\;\mathrm{Pr}[\text{$y$ is the original codeword given $z$}] \;.$$
In the BSC, it is not hard to see that the ML decoder and the MD decoder are equivalent. We note though that unlike the case of erasures, both algorithms are not efficient and in general it is not clear how to get efficient decoding algorithms for RM codes under random errors. This is quite ordinary in this are and this is also the kind of results obtained in  \cite{abbe2015reed}.
Thus, the MK (or MD) decoder can be seen as providing information theoretic decoding rather than an efficient algorithm and recovering $y$ given $z \sim \bsc{p}(y)$ is interpreted as $\mathrm{ML}(z) = y$. 

We note that the best result in this area is \cite{DBLP:journals/tit/SaptharishiSV17} (see also \cite{DBLP:conf/soda/KoppartyP18} for an improved algorithm in the same setting of parameters) that gave efficient decoding algorithms for RM codes of degree $O(\sqrt m)$ from a fraction of $1/2-o(1)$ random errors. In particular, no efficient algorithm is known for higher degrees. 

As remarked earlier, the task of recovering $y$ from $z$ (either in the BEC or in the BSC) is a probabilistic task as one cannot expect to be always correct but only with some probability. Therefore, we are interested in the probability of recovering correctly $y$ from $z \sim \bec{p}(y)$. This probability depends solely on the code's structure and $p$. 

\begin{defn}
\label{recover from errors erasures}
Let $\set{C_m}$ family of binary linear codes. We say that $\set{C_m}$ can recover from random erasures with parameter $p=p(m)$ if,
$$\mathrm{Pr}_{z \sim \bec{p}(y)}[\text{we can uniquely recover $y$ from $z$}] = 1 - o(1) \;,$$
where the $o(1)$ means that it is a function of $m$ that tends to zero as $m$ tends to infinity. Similarly, we say that $\set{C_m}$ can recover from random errors with parameter $p$ if,
$$\mathrm{Pr}_{z \sim \bsc{p}(y)}[\mathrm{ML}(z) = y] = 1 - o(1) \;.$$
\end{defn}

We now formally define the notion of capacity achieving codes for the BEC and the BSC. This becomes subtle when considering rates which are either sub-constant or approaching $1$. The following definition, which also appears in \cite{abbe2015reed}, captures the notion of achieving capacity for codes with sub-constant rates and also constant rates (but not rates approaching $1$ which we do not study in this work).
\begin{defn}
Let $\set{C_m}$ be a family of binary linear codes with rate $\set{R_m}$ and block length $\set{n_m}$. 
\begin{enumerate}
    \item 
    We say $\set{C_m}$ achieves capacity for random erasures if for any $\delta>0$ and sufficiently large $m$, the code $C_m$ can recover from random erasures with parameter $p_m = 1 - R_m(1+\delta)$. Equivalently, $C_m$ can recover from random erasures with parameter $p_m = 1-R_m(1+o(1))$ where the $o(1)$ term is a function that tends to zero with $m$.
    
    \item
    We say $\set{C_m}$ achieves capacity for random errors if for any $\delta>0$ and sufficiently large $m$, the code $C_m$ can recover from random errors with parameter $p_m$ satisfying $\binaryentropy{p_m} = 1 - R_m(1+\delta)$. Equivalently, $C_m$ can recover from random errors with parameter $p_m$ satisfying $\binaryentropy{p_m} = 1 - R_m(1+o(1))$ where the $o(1)$ term is a function that tends to zero with $m$.
\end{enumerate}
\end{defn}

\subsection{Discrete derivatives}
An important tool in our estimate of the weight distribution of RM codes is discrete derivatives.
\begin{defn}
Let $f : \F_2^m \rightarrow \F_2$ a function and $y \in \F_2^m$. Define the derivative of $f$ in direction $y$ by,
$$\derivative{y}{f}(x) = f(x+y) + f(x) \;.$$
Also, define the order $k$ derivative of $f$ in direction $Y = (y_1,\ldots,y_k)$ by,
$$\derivative{Y}{f}(x) = \derivative{y_1}{\derivative{y_2}{\cdots \derivative{y_k}{f}}}(x) \;.$$
\end{defn}

\begin{prop}
\label{derivative span lemma}
The following is true for discrete derivatives over $\F_2$.
\begin{enumerate}
\item
Degree Decrease: For any $f : \F_2^m \rightarrow \F_2$ and $y \in \F_2^m$ it holds that $\deg(\derivative{y}{f}) \leq \deg(f)-1$.

\item
The derivative is a linear operator.

\item
Commutative: 
High order derivative is independent in the order of differentiation. That is, for any $y_1,y_2 \in \F_2^m$ and $f : \F_2^m \rightarrow \F_2$ it holds that $\derivative{y_1}{\derivative{y_2}{f}} = \derivative{y_2}{\derivative{y_1}{f}}$.

\item
Let $\set{e_i}_{i=1}^{m} \subseteq \F_2^m$ be the standard basis for $\F_2^m$ then $\derivative{e_i}{f}$ is just the formal derivative of $f$ with respect to the variable $x_i$.

\item
Let $Y = (y_1,\ldots,y_k)$ and $\derivative{Y}{f}$ the $k$'th order discrete derivative of some function. If $Y$ contains linearly dependent vectors then $\derivative{Y}{f} \equiv 0$ is the zero function. Otherwise, $\derivative{Y}{f}$ depends only on $\text{span}\set{y_1,\ldots,y_k}$. This shows that the derivative in direction $Y$ actually depends only on the subspace it spans rather than the specific vectors in $Y$.

\item
Let $y_1,\ldots,y_t \in \F_2^m$ and $I = \set{j_1,\ldots, j_s} \subseteq \set{1,2,\ldots,t}$ non-empty. Then,
$$\derivative{\sum_{i \in I} y_i}{f}(x) = \sum_{\ell=1}^{s} \derivative{y_{j_\ell}}{f}\per{x+\sum_{i=1}^{\ell-1} y_{j_i}} \;.$$
\end{enumerate}
\end{prop}

\subsection{Useful inequalities}

Throughout the paper we will rely on the following well known inequalities.

\begin{thm}[Hoeffding's Inequality]\label{thm:hoeffding}
Let $X_1,\ldots,X_t$ independent random variables where each $X_i$ is supported on the interval $[a_i,b_i]$. Then,
$$\prob{}{\frac{1}{t}\sum_{i=1}^{t} X_i - \mu \geq \epsilon} \leq \exp\per{\frac{2\epsilon^2 t^2}{\sum_{i=1}^{t}(b_i-a_i)^2}} \;,$$
with $\mu = \E\left[\frac{1}{t}\sum_{i=1}^{t} X_i\right]$.
\end{thm}


\begin{thm}[Chernoff's inequality]\label{thm:chernoff}
Let $X_1,\ldots,X_n\in \set{0,1}$ independent random  variables   such that $\Pr[X_i = 1] = p$.
Then, for $0<\epsilon<1$
$$\Pr\left[\sum_i X_i \leq (1-\epsilon)pn\right] \leq \exp\per{-pn\epsilon^2/2 }\;.$$
\end{thm}

Although similar to Hoeffding's inequality, Chernoff's gives a better bound when $p$ is small.
For proofs see e.g. \cite{MU05-book}.

Another concentration inequality we shall use is McDiarmid's inequality \cite{mcdiarmid1989method}. Before stating the inequality, we need to define $L$-Lipschitz functions. We say that a function $F : \set{0,1}^n \rightarrow \R$ is $L$-Lipschitz if for all $1 \leq i \leq n$ and any choice of $x_1,\ldots,x_i,x_i',\ldots,x_n \in \set{0,1}$ it holds that
$$\abs{F(x_1,\ldots,x_i,\ldots,x_n) - F(x_1,\ldots,x_i',\ldots,x_n)} \leq L \;.$$
\begin{thm}[McDiarmid's Inequality]\label{thm:mcdiarmid}
Let $X_1,\ldots,X_n \in \set{0,1}$ be independent random variables. Let $F:\set{0,1}^n \rightarrow \R$ be $L$-Lipschitz. Then,
$$\prob{X_1,\ldots,X_n}{\abs{F(X_1,\ldots,X_n) -  \E[F(X_1,\ldots,X_n)]} \geq \epsilon} \leq \exp\per{-\frac{2\epsilon^2}{n L^2}} \;.$$
\end{thm}

Another useful inequality gives an estimate on the size of balls in the hamming metric.
\begin{lem}
\label{amir comb lem}
For any $n,k \in \N$ such that $\frac{k}{n} \leq \frac{1}{2}$ we have,
$$2^{n\binaryentropy{k/n}-O(\log n)} \leq \binom{n}{k} \leq \binom{n}{\leq k} \leq 2^{n\binaryentropy{k/n}} \;. $$
\end{lem}
A proof for the upper bound can be found in section 3.3 of \cite{guruswami2012essential} and the lower bound easily follows from Stirling's approximation $n! \approx \sqrt{2\pi n}\per{\frac{n}{e}}^n$ (for proof see e.g \cite{robbins1955remark}).

\section{Weight Distribution}
\label{section : weight distribution}

In this section we prove our main results on the weight distribution of RM codes. On the way we explain the results of \cite{kaufman2012weight} and  \cite{abbe2015reed} as we will be building upon their results. In the rest of this section we shall think of the number of variables $m$ as growing to infinity and the degree $r$ will always be $\gamma \cdot m$ for some constant $0<\gamma<1$. In some results we will need to restrict $\gamma$ to be less than $1/2$ but in some it will be unrestricted. 

\subsection{General technique of \cite{kaufman2012weight,abbe2015reed}}
We start by describing the idea behind the results of  \cite{kaufman2012weight}  and the improvements made in \cite{abbe2015reed}. 
\begin{defn}
Let $S \subseteq \polynomials{m}{r}$ be a subset of polynomials. We say $\mathcal{A} \subseteq \set{f : \F_2^m \rightarrow \F_2^m}$ is a $\delta$-net for $S$ if it satisfies the following property: For every $f \in S$ there exists $g \in \mathcal{A}$ such that $\dist{f}{g} \leq \delta$.
We stress that $\mathcal{A}$ need not to be a subset of $\polynomials{m}{r}$. 
\end{defn}

Plainly speaking, a $\delta$-net for a set $S$ is a collection of balls with radius $\delta$, in the space of functions $f : \F_2^m \rightarrow \F_2$ (i.e, all functions and not just low-degree polynomials), that covers all of $S$. The following lemma from \cite{kaufman2012weight} guarantees the existence of a small $\delta$-net for the set of polynomials $f \in \polynomials{m}{r}$ of weight at most $\beta=2^{-k}$.

\begin{lem}
[Lemma 2.2 in \cite{kaufman2012weight}]
\label{KPL1}
Let $f : \F_2^m \rightarrow \F_2$ be a function such that $\weight{f}\leq 2^{-k}$ for $k \geq 2$ and let $\delta > 0$. Then, there exist directions $Y_1,\ldots, Y_t \in (\F_2^m)^{k-1}$ such that
$$\prob{x}{f(x) \neq \text{Maj}\left(\derivative{Y_1}{f}(x),\ldots,\derivative{Y_t}{f}(x)\right)}\leq \delta \;,$$
where  $t = \lceil 17\log(1/\delta) \rceil$.
\end{lem}

We note that  \cite{kaufman2012weight} only gave the bound $t = O(\log(1/\delta))$ but to get our results we need the exact constants. We give the proof of the lemma in \Cref{constants in klp section}.

\begin{cor}
\label{delta net for low weight}
For any $k,t \in \N$ define, 
$$\mathcal{A}_{k,t} = \set{\text{Maj}\left(\derivative{Y_1}{f},\ldots,\derivative{Y_t}{f}\right) : Y_1,\ldots, Y_t \in (\F_2^m)^{k} \;,\; f \in \polynomials{m}{r}} \;.$$
Then $\mathcal{A}_{k-1,t}$ is a $\delta$-net for $\set{f \in \polynomials{m}{r} : \weight{f}\leq 2^{-k}}$ where $t = \lceil 17\log(1/\delta) \rceil$.
\end{cor}

The approach in \cite{kaufman2012weight} is to consider $\delta$ smaller than half the minimum distance of the RM code so that $\mathcal{A}$ obtained via \Cref{delta net for low weight} uniquely decodes $f \in \reedmuller{m}{r}$. Along with a simple counting argument this gives,
$$\weightdistribution{m}{r}{2^{-k}} \leq \abs{\mathcal{A}_{k-1,t}}\;,$$
where $t = O(r)$. To bound $\abs{\mathcal{A}_{k-1,t}}$, they observe that taking a discrete derivative reduces the degree hence $\abs{\mathcal{A}_{k-1,t}}$ is bounded by the number of $t$-tuples of degree $r-k+1$ polynomials. This yields,
\begin{equation}\label{eq:klp-wt}
    \weightdistribution{m}{r}{2^{-\ell}} \leq \exp_2\per{O(r) \cdot \per{\frac{\gamma}{1-\gamma}}^{\ell-1}\binom{m}{\leq r}}\;.
\end{equation}
In \cite{abbe2015reed} the authors used the following observation in order to obtain a recurrence relation for $\weightdistribution{m}{r}{2^{-k}}$.
\begin{prop}
\label{basic counting argument}
Let $S \subseteq \polynomials{m}{r}$ be a subset of polynomials with a $\delta$-net $\mathcal{A}$. Then,
$$\abs{S} \leq \abs{\mathcal{A}}\cdot \weightdistribution{m}{r}{2\delta} \;.$$
\end{prop}
\begin{proof}
By a simple counting argument, it suffices to show that every ball with radius $\delta$ in the net contains at most $\weightdistribution{m}{r}{2\delta}$ points from $S$. To see this, let $f \in \mathcal{A}$ and denote,
$$\set{g \in S : \dist{f}{g} \leq \delta} = \set{g_1,\ldots,g_j} \;.$$
We need to prove that $j \leq \weightdistribution{m}{r}{2\delta}$. Note that by the triangle inequality, 
$$\dist{g_1}{g_i} \leq 2\delta \;\;\; \forall i = 1,2,\ldots,j \;.$$ 
Therefore, the set $\set{g_1+g_i : i = 1,\ldots,j}$ contains $j$ distinct degree $r$ polynomials of weight at most $2\delta$ and so $j \leq \weightdistribution{m}{r}{2\delta}$.
\end{proof} 
\begin{cor}
\label{recursion for weight}
Let $r,m,\ell \in \N$ such that $r \leq m$. Then,
$$\weightdistribution{m}{r}{2^{-\ell}} \leq \abs{\mathcal{A}_{\ell-1,t}}\cdot\weightdistribution{m}{r}{2^{-\ell-1}}\;,$$
where $t = 17(\ell+2)$.
\end{cor}
\begin{proof}
Apply \Cref{basic counting argument} with $S = \set{f \in \polynomials{m}{r} : \weight{f} \leq 2^{-\ell}}$, $\delta = 2^{-\ell-2}$ and the net $\mathcal{A}_{\ell-1,t}$ from \Cref{delta net for low weight}.
\end{proof}

The next corollary, which follows easily from \Cref{recursion for weight}, was not stated before and relying on it simplifies some of the calculations of  \cite{abbe2015reed}.

\begin{cor}
\label{recursion for weight 2}
Let $m,\ell \in \N$, $0<\gamma<1$ and set $r = \gamma m$. Then,
$$\weightdistribution{m}{r}{2^{-\ell}} \leq \prod_{j=\ell}^{r} \abs{\mathcal{A}_{j-1,17(j+2)}} \;.$$
\end{cor}
\begin{proof}
Repeatedly apply \Cref{recursion for weight} with parameter $\ell'$ starting at $\ell'=\ell$ up to $\ell' = r+1$. When $\ell' = r+1$ we have $\weightdistribution{m}{r}{2^{-\ell'}}=\weightdistribution{m}{r}{2^{-r-1}}=1$, as the only polynomial $f \in \reedmuller{m}{r}$ of weight less than $2^{-r}$ is the zero polynomial.
\end{proof}

In \cite{abbe2015reed} this recursive approach was combined with a sharper estimate on $\abs{\mathcal{A}_{k,t}}$ (compared to the one given in \cite{kaufman2012weight}) to obtain the following improvement of \cref{eq:klp-wt}  (Theorem 1.5 in \cite{abbe2015reed}),
\begin{equation}\label{eq:asw-wt}
\weightdistribution{m}{r}{2^{-\ell}} \leq \exp_2\per{\per{O(\ell^4 \gamma^{\ell-1})+o(1)}\binom{m}{\leq r}}\;.
\end{equation}
We shall later give a better analysis (based on \Cref{recursion for weight}) and show that one can replace $O(\ell^4 \gamma^{\ell-1})$ with $O(\ell\gamma^{\ell-1})$ in \cref{eq:asw-wt}.

\subsection{Bounding \texorpdfstring{$\abs{\mathcal{A}_{k,t}}$}{Lg}}
\label{section - counting derivatives}
In this section we improve the bounds of \cite{kaufman2012weight,abbe2015reed} on the size of $\mathcal{A}_{k,t}$ (as defined in \Cref{delta net for low weight}). A naive estimate, which is the one used in \cite{kaufman2012weight}, relies on the basic observation that taking an order $k$ derivative decreases the degree by at least $k$. Therefore, one can estimate $\abs{\mathcal{A}_{k,t}}$ by the number of all possible sequences of polynomials in $m$ variables of degree $r-k$,
$$\abs{\mathcal{A}_{k,t}} \leq \exp_2\per{t\binom{m}{\leq r-k}}\;.$$
However, this estimate is far from being tight (especially for polynomials of high degrees). That is, there are much fewer polynomials of degree $r-k$ that are a derivative of order $k$ of a degree $r$ polynomial, than general degree $r-k$ polynomials. The following estimate appears in the proof of Theorem 3.3 in \cite{abbe2015reed}. As we rely on it later we shall give the proof. 
\begin{prop}[Implicit in the proof of Theorem 3.3 of \cite{abbe2015reed}]
\label{bound for derivatives from asw}
For any $k,t \in \N$ we have,
$$\abs{\mathcal{A}_{k,t}} \leq \exp_2\per{mtk + t\binom{m-k}{\leq r-k}}\;.$$
\end{prop}
\begin{proof}
Recall that
$$\mathcal{A}_{k,t} = \set{\text{Maj}\left(\derivative{Y_1}{f},\ldots,\derivative{Y_t}{f}\right) : Y_1,\ldots, Y_t \in (\F_2^m)^{k} \;,\; f \in \polynomials{m}{r}} \;.$$
Consider an order $k$ derivative in a fixed direction $Y \in (\F_2^m)^k$. Without  loss of generality, we may assume that $Y = (e_1,\ldots,e_k)$ where $e_i$ is the standard basis vector, i.e $\derivative{Y}{f}$ is simply the formal derivative according to the variables $x_1,\ldots,x_k$. It is not hard to see that $\derivative{Y}{f}$ is a degree $r-k$ polynomial in the $m-k$ variables $x_{k+1},\ldots,x_m$. Hence, there are at most $2^{\binom{m-k}{\leq r-k}}$ such polynomials. As there are at most $2^{mk}$ possible directions $Y \in \per{\F_2^m}^k$, there are at most $2^{mk+\binom{m-k}{\leq r-k}}$ polynomials of the form $\derivative{Y}{f}$. In particular, the number of sequences $(\derivative{Y_1}{f},\ldots,\derivative{Y_t}{f})$ is at most $ \per{2^{mk+\binom{m-k}{\leq r-k}}}^t$. Clearly, this is also an upper bound on $\abs{\mathcal{A}_{k,t}}$.
\end{proof}

We shall now present an additional saving which applies to $\mathcal{A}_{1,t}$ and is essential for the analysis in \Cref{section - upper bound bias}. Any $F \in \mathcal{A}_{1,t}$ is a function of the following form,
$$F = \text{Maj}\set{\derivative{y_1}{f}(x),\derivative{y_2}{f}(x), \ldots, \derivative{y_t}{f}(x)}$$
for some $f \in \reedmuller{m}{r}$ and $y_1,\ldots,y_t \in \F_2^m$. Therefore, $F$ is completely determined by specifying the directions $y_1,\ldots,y_t$ and the derivatives $(\derivative{y_1}{f},\derivative{y_2}{f}, \ldots, \derivative{y_t}{f})$. The key observation is that specifying $\derivative{y_1}{f}$ already provides a lot of information about $f$ and in particular about its derivatives. Intuitively, given $\derivative{y_1}{f}$ there are significantly fewer possible values for $\derivative{y_2}{f}$ than in the case where we do not know $\derivative{y_1}{f}$. While \Cref{bound for derivatives from asw} is proved by estimating the number of distinct possible functions $\derivative{y}{f}$ (for a fixed $y$), the proof of our next proposition upper bounds the number of distinct possible sequences,
$$(\derivative{y_1}{f},\derivative{y_2}{f}, \ldots, \derivative{y_t}{f})$$
for fixed $y_1,\ldots,y_t \in \F_2^m$.
\begin{prop}
\label{bound on a1t}
Let $m,r, t \in \N$ such that $t,r \leq m$ and write $\gamma = r/m$. Then,
$$\abs{\mathcal{A}_{1,t}} \leq \exp_2\per{mt + \sum_{j=1}^{t}\binom{m-j}{\leq r-1}}  \leq  \exp_2\per{mt +\per{1 - \per{1-\Tilde{\gamma}}^t} \binom{m}{\leq r}}\;,$$
where $\Tilde{\gamma} = \gamma\left(1+\frac{t}{m-t}\right)$.
\end{prop}
\begin{proof}
Fix some directions $y_1,\ldots,y_t \in \F_2^m$. First note that we may assume $y_1,\ldots,y_t$ are linearly independent as if $y_j$ is some combination of $\set{y_i}_{i \in I}$ then $\derivative{y_j}{f}$ is completely determined by $\set{\derivative{y_i}{f}}_{i \in I}$ (See \Cref{derivative span lemma}). For simplicity, after applying a linear transformation, we may also assume without loss of generality that $y_i = e_i$ and so $\derivative{y_i}{f}$ is just the formal derivative with respect to $x_i$. Therefore the sequence,
$$(\derivative{y_1}{f},\derivative{y_2}{f}, \ldots, \derivative{y_t}{f})$$
is determined only by the monomials of $f$ containing $x_i$ for some $i = 1, \ldots, t$. Thus, if we count the number of monomials containing $x_1$, then those that contain $x_2$ but not $x_1$ etc. we get that there are exactly,
$$\sum_{j=1}^{t}\binom{m-j}{\leq r-1}$$
such monomials. Hence, there are at most $\exp_2\per{\sum_{j=1}^{t}\binom{m-j}{\leq r-1}}$ such distinct sequences. This estimate holds for fixed directions $y_1,\ldots,y_t \in \F_2^m$. In order to get an upper bound on $\abs{\mathcal{A}_{1,t}}$, we need to take the union over all directions which gives another factor of $2^{mt}$.

To get the upper bound estimate observe the following combinatorial identity,
$$\sum_{j=1}^{t}\binom{m-j}{\leq r-1} = \binom{m}{\leq r} - \binom{m-t}{\leq r} \;.$$
A combinatorial explanation to this identity is that both sides count all degree $r$ monomials that contains some $x_i$ for $i = 1,2,\ldots,t$ (an algebraic proof can be obtained using the identity $\binom{n}{\leq k} = \binom{n-1}{\leq k-1} + \binom{n-1}{\leq k}$). 
We shall also need the following inequality whose proof is given in \Cref{Combinatorial Lemmas section}.
\begin{lem}
\label{(Simple Combinatorial Bound III)}
Let $t,r\leq m \in \N$. and write $\gamma = r/m$. Then, for  $\Tilde{\gamma} = \gamma\left(1+\frac{t}{m-t}\right)$ it holds that
$$\binom{m-t}{\leq r} \geq \per{1-\Tilde{\gamma}}^t\binom{m}{\leq r} \;.$$
\end{lem}
Combining the two identities we get
$$ \sum_{j=1}^{t}\binom{m-j}{\leq r-1} = \binom{m}{\leq r} - \binom{m-t}{\leq r} \leq \binom{m}{\leq r}-\binom{m}{\leq r}\per{1-\Tilde{\gamma}}^t =  \binom{m}{\leq r}\per{1 - \per{1-\Tilde{\gamma}}^t} \;,$$
and the upper bound follows.
\end{proof}

\begin{remark}
One can try to make a similar argument also for sequences of higher order derivatives. However, this is less easy than in the case of order-$1$ derivatives. 
The main issue is that in the proof of \Cref{bound on a1t} we used one basis and it was clear how different derivatives contribute  to one another. In contrast, higher order derivatives may be with respect to very different subspaces that do not necessarily exhibit any apparent structure that can be used to quantify the mutual contributions. 
We note however that even if a similar calculation could be obtained then the best bound one can hope to get from such an argument is 
$$\abs{\mathcal{A}_{k,t}} \leq \exp_2 \per{(1-(1-\gamma^k)^t) \binom{m}{\leq r}} \;.$$
Indeed, assuming $kt\leq m$,  the following sequence of derivatives for $j=1\ldots t$: 
$$Y_j = (e_{jk+1}, \ldots, e_{(j+1)k}) \;,$$
would give that bound on  $\abs{\mathcal{A}_{k,t}}$
We note though that even assuming that this is the extremal case, one will not obtain a significant improvement to \Cref{main thm - low weight}.
\end{remark}

\subsection{Weight distribution for small weights}
\label{section - upper bound weight}
This section includes a new upper bound on $\weightdistribution{m}{r}{\epsilon}$. It is obtained by following the approach of \cite{abbe2015reed} using the upper bound on $\mathcal{A}_{k,t}$ given in \Cref{bound for derivatives from asw}.

\begin{customthm}{\ref{main thm - low weight}}
Let $r, m,\ell \in \N$ such that $r \leq m$ and write $\gamma = r/m$. Then,
$$W_{m,r} (2^{-\ell})\leq \exp_2\per{O(m^4) + 17(c_{\gamma}\ell+d_{\gamma})\gamma^{\ell-1}\binom{m}{\leq r}} \;,$$
where $c_{\gamma} = \frac{1}{1-\gamma}$ and $d_{\gamma} = \frac{2-\gamma}{(1-\gamma)^2}$.
\end{customthm}
\begin{proof}
Apply \Cref{recursion for weight 2} and the bound on $\abs{\mathcal{A}_{k,t}}$ in \Cref{bound for derivatives from asw} to obtain,
\begin{align}
\nonumber
W_{m,r}(2^{-\ell}) &\leq \prod_{j=\ell}^{r} \abs{\mathcal{A}_{j-1,17(j+2)}}\\
\nonumber
&\leq \exp_2\per{\sum_{j=\ell}^{r} 17 m(j-1)(j+2) + 17(j+2)\binom{m-(j-1)}{\leq r-(j-1)}}\\
\nonumber
&\leq^{(*)} \exp_2\per{O(m^4) + 17\sum_{j=\ell}^{r} (j+2)\gamma^{j-1}\binom{m}{\leq r}}\\
\nonumber
&= \exp_2\per{O(m^4) + 17\sum_{j=0}^{r-\ell} (\ell + j+2)\gamma^{\ell + j-1}\binom{m}{\leq r}}\\
\nonumber
&\leq \exp_2\per{O(m^4) + 17\gamma^{\ell-1}\binom{m}{\leq r}\per{\ell\sum_{j=0}^{\infty}\gamma^{j}+\sum_{j=0}^{\infty}(j+2)\gamma^{j}}}\\
\nonumber
&\leq \exp_2\per{O(m^4) + 17(c_{\gamma}\ell+d_{\gamma})\gamma^{\ell-1}\binom{m}{\leq r}} \;,
\end{align}
where $c_{\gamma} = \frac{1}{1-\gamma}$ and $d_{\gamma} = \frac{2-\gamma}{(1-\gamma)^2}$. Inequality (*)
follows from the inequality $\binom{m-j}{\leq r-j} \leq \per{\frac{r}{m}}^j \binom{m}{\leq r}$ (see \Cref{(Simple Combinatorial Bound I)} in the appendix).
\end{proof}

\subsection{Weight distribution for small bias}
\label{section - upper bound bias}
This section includes two upper bounds on the number of polynomials $f \in \polynomials{m}{r}$ with $\bias{f} \geq \epsilon$. The first upper bound follows the ideas used to prove \Cref{main thm - low weight}. Unfortunately it only applies to degrees $r = \gamma m$ where $\gamma < 1/2$. The second upper bound is obtained using an elementary concentration inequality known as McDiarmid's inequality and holds for any $\gamma > 0$. We remark that the second upper bound is much weaker than the first and it is instructive to compare the two.

\subsubsection{Upper bound using derivatives}
The starting point is Lemma 2.4 in \cite{kaufman2012weight}, which is analogous to \Cref{KPL1} for low bias polynomials, and guarantees the existence of a small $\delta$-net for low-bias polynomials. 
\begin{lem}
[Lemma 2.4 in \cite{kaufman2012weight}]
\label{KPL2}
Let $f : \F_2^n \rightarrow \F_2$ be a function such that $\bias{f} \geq \epsilon > 0$ and let $\delta>0$. Then, for  $t = \lceil 2\log(1/\epsilon) + \log(1/\delta) + 1 \rceil$, there exist directions $y_1,\ldots,y_t \in \F_2^m$ such that,
$$\prob{x}{f(x) = \text{Maj}\left(\derivative{\sum_{i \in I}y_i}{f}(x) : \emptyset \neq I \subseteq [t]\right)} \geq 1 - \delta \;.$$
\end{lem}
We give the proof of the lemma  in  \Cref{constants in klp section}.
\begin{cor}\label{cor:def-bt}
For any $t \in \N$ define,
$$\mathcal{B}_{t} = \set{\text{Maj}\left(\derivative{\sum_{i \in I}y_i}{f}(x) : \emptyset \neq I \subseteq [t]\right) : f \in \reedmuller{m}{r} \;,\; y_1,\ldots,y_t \in \F_2^m } \;.$$
Then, for  $t = \lceil 2\log(1/\epsilon) + \log(1/\delta) + 1 \rceil$, $\mathcal{B}_{t}$ is a $\delta$-net for $\set{f \in \reedmuller{m}{r} : \bias{f} \geq \epsilon}$.
\end{cor}
\begin{cor}\label{cor:bias-via-net}
Let $r,m, s, \ell \in \N$ such that $r \leq m$. Set $t=2\ell+s+1$. Then,
$$\weightdistribution{m}{r}{\frac{1-2^{-\ell}}{2}} \leq \abs{\mathcal{B}_t}\cdot \weightdistribution{m}{r}{2^{-s+1}} \;.$$
\end{cor}
\begin{proof}
This follows from combining \Cref{basic counting argument} with \Cref{cor:def-bt} for $\epsilon=2^{-\ell}$ and $\delta = 2^{-s}$.
\end{proof}

\begin{prop}
\label{sharp estimate for small bias net}
Let $t,r\leq m\in \N$, ${\mathcal{B}_{t}}$ as in \Cref{cor:def-bt} and set $\gamma = r/m$. Then,
$$\abs{\mathcal{B}_{t}}\leq \exp_2\per{mt + \sum_{j=1}^{t}\binom{m-j}{\leq r-1}}  \leq  \exp_2{\per{mt +(1-(1-\Tilde{\gamma})^t )\binom{m}{\leq r}}} \;,$$
where $\Tilde{\gamma} =  \gamma\left(1+\frac{t}{m-t}\right)$.
\end{prop}
\begin{proof}
By \Cref{derivative span lemma} first order derivatives in directions $y_1,\ldots,y_t$ determines the derivatives in every direction within $\text{span}\set{y_1,\ldots,y_t}$. Hence, $\abs{\mathcal{B}_{t}} \leq \abs{\mathcal{A}_{1,t}}$ and the bound follows from \Cref{bound on a1t}.
\end{proof}

\begin{prop}
\label{small bias estimation sharper}
Let $m, r ,s, \ell \in \N$ such that $r \leq m$ and write $\gamma = r/m$. Then,
$$\weightdistribution{m}{r}{\frac{1-2^{-\ell}}{2}} \leq \exp_2\per{O(m^4)+\per{1-(1-\Tilde{\gamma})^{2\ell+s+1} + 17(c_{\gamma}(s-1)+d_{\gamma}) \gamma^{s-2}}\cdot \binom{m}{\leq r}} \;,$$
where $\Tilde{\gamma} = \gamma\left(1+\frac{2\ell + s + 1}{m-(2\ell + s + 1)}\right)$, $c_{\gamma} = \frac{1}{1-\gamma}$, $d_{\gamma} = \frac{2-\gamma}{(1-\gamma)^2}$.
\end{prop}
\begin{proof}
The proof follows from combining \Cref{cor:bias-via-net} with the estimates in \Cref{sharp estimate for small bias net} and \Cref{main thm - low weight}.
%
\end{proof}

We are now ready to prove our main estimate on the weight distribution of RM codes. 

\begin{customthm}{\ref{main thm - low bias estimation}}
Let $\ell,m,\in \N$ and let $0 < \gamma(m) < 1/2 - \Omega\per{\sqrt{\frac{\log m}{m}}}$ be a parameter (which may be constant or depend on $m$) such that $\frac{\ell+\log\frac{1}{1-2\gamma}}{(1-2\gamma)^2} = o(m)$. Then, 
$$\weightdistribution{m}{\gamma m}{\frac{1-2^{-\ell}}{2}} \leq \exp_2\per{O(m^4)+\per{1-2^{-c(\gamma,\ell)}}\binom{m}{\leq r}} \;,$$
where  $c(\gamma,\ell) = O\per{\frac{\gamma^2 \ell  +  \gamma \log(1/1-2\gamma)}{1-2\gamma} + \gamma}$.
\end{customthm}
\begin{proof}
Let 
$s=s(\gamma,\ell)$ be the smallest natural number for which the following holds,
$$17 (2s+4)\gamma^{s-2} \leq \frac{1}{2}\left(1-\gamma\left(1+\frac{2\ell + s + 1}{m-(2\ell + s + 1)}\right)\right)^{2\ell+s+1} \;.$$
It is not hard to see that\footnote{It is possible to get more accurate bounds on $s$ but since it does not play a major role in our proofs we settle for the more rough estimate.} $s =s(\gamma,\ell)= O\per{\frac{\gamma\ell + \log(1/1-2\gamma)}{1-2\gamma}}$. A short calculation that justifies this estimate appears in \Cref{section: small calculation}. We remark that the requirement $\frac{\ell+\log\frac{1}{1-2\gamma}}{(1-2\gamma)^2} = o(m)$ enables us to effectively replace $\tilde{\gamma}$ by $\gamma$ (See \Cref{section: small calculation} for details).
Let $t=2\ell + s + 1$.
Applying \Cref{small bias estimation sharper} we get, for $\tilde{\gamma} =  \gamma\left(1+\frac{2\ell + s + 1}{m-(2\ell + s + 1)}\right)$, that
\begin{align}
\nonumber
\weightdistribution{m}{\gamma m}{\frac{1-\epsilon}{2}}  &\leq \exp_2\per{O(m^4)+\per{1-(1-\tilde{\gamma})^{2\ell+s+1} + 17\per{\frac{s-1}{1-\gamma}+\frac{2-\gamma}{(1-\gamma)^2}} \gamma^{s-2}}\cdot \binom{m}{\leq \gamma m}}\\
\nonumber
&\leq^{(*)} \exp_2\per{O(m^4)+\per{1-(1-\tilde{\gamma})^{2\ell+s+1} + 17(2s+4) \gamma^{s-2}}\cdot \binom{m}{\leq \gamma m}}\\
\nonumber
&\leq^{(\dagger)} \exp_2\per{O(m^4)+\per{1-1/2(1-\tilde{\gamma})^{2\ell+s+1}}\cdot \binom{m}{\leq \gamma m}}\\
\nonumber
&\leq \exp_2\per{O(m^4)+\left(1-2^{-c(\gamma,\ell)}\right)\binom{m}{\leq \gamma m}} \;,
\end{align}
where $c(\gamma,\ell) = \log(1/(1-\tilde{\gamma}))\cdot \per{2\ell+s+1} = O\per{\frac{\gamma^2 \ell  +  \gamma \log(1/1-2\gamma)}{1-2\gamma}}$.
Note that inequality ${(*)}$ holds since $\gamma < 1/2$ and inequality ${(\dagger)}$ is due to the choice of $s$.
\end{proof}

We note that the proof crucially relied on $\gamma<1/2$ as otherwise we have that
$$\per{1-(1-\tilde{\gamma})^{2\ell+s+1} + 17\per{\frac{s-1}{1-\gamma}+\frac{2-\gamma}{(1-\gamma)^2}}\gamma^{s-2} } > 1$$
and we do not get a meaningful upper bound on $\weightdistribution{m}{r}{\frac{1-\epsilon}{2}}$.

\subsubsection{Upper bound for high degrees}
\Cref{main thm - low bias estimation} proved upper bound on the number of polynomials having small bias for $\gamma<1/2$. In this section we given an upper bound on $\weightdistribution{m}{r}{\frac{1-2^{-\ell}}{2}}$ for any $r$. The estimate that we get is weaker than the one in \Cref{main thm - low bias estimation}, but its advantage is that it works for all degrees. 
%
%

We are now ready to prove \Cref{thm : weak concentration of bias for all gamma}. To ease the reading we repeat its statement.
\begin{customthm}{\ref{thm : weak concentration of bias for all gamma}}
Let $r \leq m\in \N$ and $\epsilon > 0$. Then,
$$\prob{f \sim \reedmuller{m}{r}}{\abs{\bias{f}} > \epsilon} \leq 2\exp\per{-\frac{2^{r}\epsilon^2}{2}} \;.$$
\end{customthm}
\begin{proof}
First note that $\E_{f \in \reedmuller{m}{r}}[\bias{f}] = 0$. Althought this is simple, for a reason that will soon be clear, we show this using the following symmetry argument: The mapping $f \rightarrow 1 + f$ is a bijection from $\reedmuller{m}{r}$ to itself and $\bias{1+f} = -\bias{f}$. Hence, every contribution of $f \in \reedmuller{m}{r}$ to the expectation $\E_{f \in \reedmuller{m}{r}}[\bias{f}]$ is cancelled by $1+f$ and so $\E_{f \in \reedmuller{m}{r}}[\bias{f}] = 0$.

In order to avoid cumbersome notation set $M = \binom{m}{\leq r}$. We would like to think of the bias  as a function  from the boolean hypercube $\{0,1\}^{M}$ to $\R$. Formally, let $\mathcal{B} = \set{h_i(x)}_{i=1}^{M}$ be a basis to $\reedmuller{m}{r}$ and define, $F:\set{0,1}^M  \rightarrow \R$ as
$$F(a_1,\ldots,a_M) = \mathrm{bias}\per{\sum_{i=1}^{M} a_i h_i(x)}\;.$$
Note that,
$$\prob{a_i \sim \set{0,1}}{F(a_1,\ldots,a_M) > \epsilon} = \prob{f \sim \reedmuller{m}{r}}{{\bias{f}} > \epsilon} \;.$$
We would like to apply McDiarmid's inequality (\Cref{thm:mcdiarmid}) for $F$. Thus, we would like to show that $F$ is a Lipschitz functions. This is not clear however and in fact, the question of whether $F$ is Lipschitz depends on the chosen basis. We next exhibit a basis for which $F$ satisfies the Lipschitz condition.

Let $$S = \set{\prod_{i=1}^{r} (x_i + b_i) : b_i \in \set{0,1}} \subset \reedmuller{m}{r} \;.$$ Clearly $|S|=2^{r}$.  Also, note that $\sum_{h \in S} h \equiv 1$ the constant function.
It is not hard to see that the elements of $S$ are linearly independent and hence can be completed to a basis of $\reedmuller{m}{r}$. Denote this basis by $\mathcal{B} = \set{h_i(x)}_{i=1}^{M}$ and assume $S = \set{h_{M-2^{r} +1},\ldots, h_M}$, i.e. we order the elements in $\mathcal{B}$ in such a way that the elements of $S$ are the last $2^{r}$ basis elements.

Given $f \in \reedmuller{m}{r}$, let $f = \sum_{i=1}^{M} c_i(f) h_i(x)$ where $c_i(f) \in \set{0,1}$ are the coefficients of $f$ with respect to the basis $\mathcal{B}$. Partition $\reedmuller{m}{r}$ to subsets as follows,
$$Z_{t_1,\ldots,t_{M-2^{\gamma m}}} = \set{f \in \reedmuller{m}{r} : c_i(f) = t_i \; \forall 1\leq i \leq M - 2^{r}}\;.$$
Namely, we fix all coefficients of basis elements that do not belong to $S$ to some values. By the law of total probability,
$$\prob{f}{\abs{\bias{f}} > \epsilon} = \sum_{t_1,\ldots,t_{M-2^{r}}} 2^{-(M-2^{r})} \prob{f}{\abs{\bias{f}}> \epsilon \; | \;  f \in Z_{t_1,\ldots,t_{M-2^{r}}}}\;.$$
Thus, it suffices to prove that for every $t_1,\ldots, t_{M-2^{r}}$ it holds that,
$$\prob{f}{\abs{\bias{f}} > \epsilon \; | \; f \in Z_{t_1,\ldots,t_{M-2^{r}}}} \leq 2\exp\per{-\frac{2^{r}\epsilon^2}{2}}$$
Fix $t_1,\ldots, t_{M-2^{r}} \in \set{0,1}$ and set $Z = Z_{t_1,\ldots, t_{M-2^{r}}}$. We first note that the restriction of $F$ to $Z$, denoted $F\big|_Z$, is $2^{1-r}$-Lipschitz. To see this notice that $\weight{h_i} = 2^{-r}$. Thus, the difference in the bias of a function $f$ and $f+h_i$ is at most twice the weight of $h_i$ which is at most $2^{1-r}$. Since $\sum_{h \in S} h \equiv 1$  the mapping $f \rightarrow 1 + f$ is a bijection from $Z$ to itself. Using the same symmetry argument used to show that $\E_{f}[\bias{f}] = 0$ we deduce that $\E[F\big|_Z] = 0$. Applying McDiarmid's inequality we get that
$$\prob{a_i \sim \set{0,1}}{\abs{F\big|_Z(a_{M-2^{r} + 1}, \ldots,a_{M})} \geq \epsilon} \leq 2\exp\per{-\frac{2 \epsilon^2}{2^{r}\cdot 2^{2(1-r)}}}=2\exp\per{-\frac{2^{r}\epsilon^2}{2}} \;,$$
and since $\prob{a_i \sim \set{0,1}}{\abs{F\big|_Z(a_{M-2^{r} + 1}, \ldots,a_{M})} \geq \epsilon} = \prob{f}{\abs{\bias{f}} \geq \epsilon \; | \; f \in Z}$ the claim follows.
\end{proof}

\subsection{Lower bound on the weight distribution}
\label{section : lower bound on wd}
In this section we prove a lower bound on the number of polynomials that have bias at least $\epsilon$. To avoid the use of ceilings and floors we shall prove the result for bias of the form $\epsilon=2^{-\ell}$.

\begin{customthm}{\ref{main thm: lower bound for bias}}
Let $20\leq r \leq m,\in \N$. Then for any integer $\ell <r/3$ and sufficiently large $m$ it holds that
$$\abs{f \in \reedmuller{m}{r} : \bias{f} \geq 2^{-\ell}} \geq \frac{1}{2}\cdot \exp_2\per{\sum_{j=1}^{\ell-1}\binom{m-j}{\leq r-1}} \;.$$
\end{customthm}
\begin{proof}
Consider the following random polynomial,
$$g(x_1,\ldots,x_m) = \sum_{i=1}^{\ell} x_i f_i(x_{i+1},\ldots,x_{m})\;,$$
where $f_i \sim \polynomials{m-i}{r-1}$ uniformly at random. It is not hard to see that different choices of $(f_1,\ldots,f_\ell)$ yield different polynomials $g(x) \in \polynomials{m}{r}$. We will show that with probability at least $1/2$, over the choice of $g$, it holds that $\bias{g}\geq 2^{-\ell+1}$. 
As there are $\exp_2\per{\sum_{j=1}^{\ell}\binom{m-j}{\leq r-1}}$  such different polynomials $g$, the lower bound follows. For $(a_1,\ldots,a_\ell) \in \F_2^\ell$ define
$$g\big|_{(a_1,\ldots,a_\ell)}(x_{\ell+1},\ldots,x_{m}) = \sum_{i=1}^{\ell} a_i f_i(a_{i+1},\ldots,a_\ell,x_{\ell+1},\ldots, x_{m})\;.$$
Note that
\begin{align}\label{eq:g}
\bias{g} &= 2^{-\ell} + \E_{(a_1,\ldots,a_\ell)\neq (0,\ldots,0)}[\bias{g\big|_{(a_1,\ldots,a_\ell)}}] \;,
\end{align}
where the $2^{-\ell}$ term comes from the probability that $ (a_1,\ldots,a_\ell)  = (0,\ldots,0)$. Next we show that with good probability the second term in the RHS is at least $-2^{\ell+1}$. 

\sloppy Fix some $ (a_1,\ldots,a_\ell)  \neq (0,\ldots,0)$ and observe that $g\big|_{(a_1,\ldots,a_\ell)}$ is a uniformly random polynomial, over the variables $x_{\ell+1},\ldots,x_m$, of degree at most $r-1$. Indeed, let $k$ be such that $a_k =1$. Then, as $$g\big|_{(a_1,\ldots,a_\ell)}(x_{\ell+1},\ldots,x_{m})= f_k(a_{k+1},\ldots,a_\ell,x_{\ell+1},\ldots,x_{m})+\sum_{i=1,i\neq k}^{\ell} a_i f_i(a_{i+1},\ldots,a_\ell,x_{\ell+1},\ldots,x_{m})$$ and $f_k(a_{k+1},\ldots,a_\ell,x_{\ell+1},\ldots,x_{m})$ is a uniformly random polynomial of degree $r-1$ in $(x_{\ell+1},\ldots,x_m)$, we get that so is $g\big|_{(a_1,\ldots,a_\ell)}$. Using \Cref{thm : weak concentration of bias for all gamma} we have
$$\mathrm{Pr}\left[\abs{\bias{g\big|_{(a_1,\ldots,a_\ell)}}} \geq 2^{-\ell-1} \right] \leq 2\exp\per{-\frac{2^{r-1}2^{-2\ell-2} }{2}} \;.$$
By union bound we get that with probability at least 
$$1 - 2^{\ell}\cdot2\exp\per{-\frac{2^{r-1}2^{-2\ell-2} }{2}} > 1/2 \;,$$ 
it holds that $\bias{g\big|_{(a_1,\ldots,a_\ell)}}> -2^{-\ell-1}$ for every $(a_1,\ldots,a_t) \neq (0,\ldots,0)$. Hence,
$$\mathrm{Pr}\left[\bias{g} \geq 2^{-\ell-1}\right]  \geq \mathrm{Pr}\left[\bias{g\big|_{(a_1,\ldots,a_\ell)}} \leq -2^{-\ell-1} \;\; \forall (a_1,\ldots,a_\ell) \neq (0,\ldots,0) \right] > 1/2\;.$$
This completes the proof.
\end{proof}

\begin{remark}\label{rem:comparison}
We would like to compare the lower bound in \Cref{main thm: lower bound for bias} with the upper bound in \Cref{main thm - low bias estimation}. These two estimates seem very different, but diving into the proof of \Cref{main thm - low bias estimation} we see that the lower bound is obtained via \Cref{cor:bias-via-net}. The upper bound on $\abs{\mathcal{B}_{t}}$ is given in \Cref{sharp estimate for small bias net}. Thus, the lower bound in \Cref{main thm - low bias estimation} has the leading term $$\sum_{j=1}^{\ell-1}\binom{m-j}{\leq r-1}$$ in the exponent whereas the upper bound has as leading term the sum $$\sum_{j=1}^{t}\binom{m-j}{\leq r-1}\;,$$ where $t$ as calculated in the proof of  \Cref{main thm - low bias estimation} is at least $2\ell$. Thus, there is at least a gap of a quadratic factor between the lower and upper bounds. That is, our lower bound on the number of polynomials that have bias at least $\epsilon$ has roughly the same leading term as the upper bound on the number of polynomials that have bias at least $\sqrt{\epsilon}$.\end{remark}


\section{Reed-Muller codes under random erasures}
In this section we study the behavior of RM codes under random erasures. In \Cref{section: rm achieve capacity bec} we prove \Cref{main thm - capacity for bec} showing that RM codes achieve capacity for the BEC for degrees at most $m/50$. In  \Cref{sec:RM-erasures-1/2} we prove \Cref{main thm - noise for bec} showing that for degrees up to $(1/2 - o(1))m$ (for some explicit $o(1)$ function) RM codes can recover from $1-o(1)$ random erasures.

Throughout this section we denote $r=\gamma m$ where $\gamma(m) \in (0,1/2)<1/2$ is some parameter that can be a constant or a function of $m$. We denote the probability that the family $\reedmuller{m}{\gamma m}$ cannot recover from random erasures with parameter $p=p(m,\gamma)$ with $\lambda_{\mathrm{BEC}}(p,\gamma)$ (See \Cref{recover from errors erasures}).

We start by proving an upper bound on $\lambda_{\mathrm{BEC}}(p,\gamma)$ that we shall use in both proofs. Recall that the $\bec{p}$ erases every coordinate with probability $p$ independently.  I.e. every coordinate is replaced with the symbol `?' with probability $p$. Since codewords of the RM code correspond to evaluation vectors, we can view corrupted codewords as evaluation vectors where some of the evaluations were erased. We shall refer to the set of evaluation points erased from the codeword as the erasure pattern. Thus, given a codeword $f \in \reedmuller{m}{r}$ and an erasure pattern $S \subseteq \F_2^m$, the corresponding corrupted codeword is the evaluation vector of $f$ with the evaluations over the set $S$ erased. The following well known lemma states exactly when can an erasure pattern be fixed. In particular, this property only depends on the erasure pattern and not on the codeword whose evaluations were erased. 
\begin{lem}
\label{recover from erasure equivalent to support of codewords}
Let $f \in \polynomials{m}{r}$ be a codeword and suppose we erase the evaluations on a set $S \subseteq \F_2^m$. Then, we can uniquely recover $f$ iff there is no nonzero $g \in \polynomials{m}{r}$ satisfying $\support{g} \subseteq S$, where $\support{g} = \set{x \in \F_2^m : g(x) \neq 0}$.
\end{lem}
\begin{proof}
We cannot uniquely decode $f$ iff there exists $h \neq f$ such that $f\big|_{\overline{S}} = h\big|_{\overline{S}}$ or equivalently $(f-h)\big|_{\overline{S}} = 0$. By linearity of the code $\reedmuller{m}{r}$ the lemma follows. 
\end{proof}

Next, we use \Cref{recover from erasure equivalent to support of codewords} to give an upper bound on $\lambda_{\mathrm{BEC}}(p,\gamma)$.
\begin{lem}
\label{probability of recovering from random erasures union bound}
For $ m \in \N$. Let $\gamma =\gamma(m)\in (0,1/2)$ be some parameter. Set $r=\gamma m$ and denote $R$ the rate of the RM code $\reedmuller{m}{r})$. For a parameter $c = c(m) \geq 1$ which may be constant or a function of $m$ let $ p_c = 1 - c \cdot R$. Then for large enough $m$ it holds that,
$$\lambda_{\mathrm{BEC}}(p_c,\gamma) \leq \mu(m,\gamma,c) + \sum_{0 \neq f \in \polynomials{m}{r}} (1-\weight{f})^{c(1 - o(1))\binom{m}{\leq r}} \;,$$
where $\mu(m,\gamma,c) = \exp\per{-\Omega\per{c\cdot \binaryentropy{\gamma}m-O(\log m)}}$.
\end{lem}
\begin{proof}
Since $c$ is fixed throughout the proof we drop the subscript $c$ and denote $p=p_c$. Let $y \in \reedmuller{m}{r}$ and $z \sim \bec{p}(y)$. Also, to avoid cumbersome notation we shall write $c$ instead of $c(m)$. Denote by $S=S(y,z) \subseteq \F_2^m$ the corresponding erasure pattern of $z$ and so $z_i = y_i$ for $i \not \in S$ and $z_i = `?'$ for $i \in S$. Note that $S$ is chosen at random such that $i \in S$ with probability $p$ and $i \not \in S$ with probability $1-p$ independently for every $i \in \F_2^m$. In order to keep the notation simple, we shall denote the bad event in which we cannot recover from the erasure pattern $S$ by $\mathcal{B}$. Thus we need to show that,
$$\lambda_{\mathrm{BEC}}(p,\gamma) =\prob{S}{\mathcal{B}} \leq \mu(m,\gamma,c) + \sum_{0 \neq f \in \polynomials{m}{r}} (1-\weight{f})^{c(1 - o(1))\binom{m}{\leq r}} \;.$$
It is more convenient to consider the set of non-erased points, that is $\overline{S}$. Typically, there are $(1-p)$ fraction of non-erased points (i.e, $\abs{\overline{S}} \approx (1-p) 2^m$). We first show, using Chrenoff's bound, that with high probability the number of erasures is not much larger than $p2^m$. We then  bound the error in when there are not too many erasures. Let $\epsilon > 0$ be a parameter which we shall determine later and consider the following two events: 
\begin{itemize}
    \item 
    The number of non-erased points is not typical, namely $\abs{\overline{S}} < (1-\epsilon)(1-p)2^m$. Denote this event by $\mathcal{A}_1$.
    
    \item
    The number of non-erased points is typical, namely $\abs{\overline{S}} \geq (1-\epsilon)(1-p)2^m$ but we cannot recover from the erasure pattern $S$. Denote this event by $\mathcal{A}_2$.
\end{itemize}

By union bound, 
$$\prob{S}{\mathcal{B}} \leq \prob{}{\mathcal{A}_1} + \prob{}{\mathcal{A}_2}\;.$$

We start by handling the event $\mathcal{A}_1$. By Chernoff's inequality (\Cref{thm:chernoff}) we get that
$$\Pr[\mathcal{A}_1] \leq \exp\per{-\frac{1}{2}\epsilon^2(1-p)2^m} = \exp\per{-\Omega\per{c\epsilon^2 2^{\binaryentropy{\gamma}m - O(\log m)}}}\;,$$ 
where the last equality holds as $1-p = c\cdot R$ and by \Cref{amir comb lem} $R \geq 2^{(\binaryentropy{\gamma}-1)m-O(\log m)}$. 
This alone imposes a ``largeness'' condition on $\epsilon$, namely, we must have $\epsilon = \omega(2^{-\binaryentropy{\gamma}m/2 + O(\log m)})$ to get   $\Pr[\mathcal{A}_1] = o(1)$. Set $\epsilon = 2^{-\frac{\binaryentropy{\gamma}}{4}m}$ then $\epsilon = o(1)$ and  
$$\exp\per{-\Omega\per{c\epsilon^2 2^{\binaryentropy{\gamma}m - O(\log m)}}} = \mu(m,\gamma,c)\;.$$
We next bound $\mathcal{A}_2$. We start by conditioning on the number of non-erasures. It will be convenient to represent the fraction of non-erasure with $\nu$. I.e $\abs{\overline{S}} = \nu 2^m$. Since we are in the case where there were not too many erasures we have that $\nu\geq (1-\epsilon)(1-p)2^m$.
$$\prob{}{\mathcal{A}_2} = \sum_{\nu\geq (1-\epsilon)(1-p)2^m} \prob{}{\abs{S} = \nu2^m} \cdot \prob{}{\mathcal{B} | \abs{\overline{S}} = \nu 2^m}\;.$$
Clearly, the probability $\prob{}{\mathcal{B} \; | \; |\overline{S}| = \nu 2^m}$ gets larger as $\nu$ gets smaller hence,
\begin{align*}
\prob{}{\mathcal{A}_2} &= \sum_{\nu \geq (1-\epsilon)(1-p)2^m} \prob{S}{|{\overline{S}}| = \nu 2^m} \cdot \prob{}{\mathcal{B} | |{\overline{S}}| = \nu 2^m} \\
&\leq \sum_{\nu \geq (1-\epsilon)(1-p)2^m} \prob{}{|\overline{S}| = \nu 2^m} \cdot \prob{}{\mathcal{B} \; | \; |\overline{S}| = (1-p)(1-\epsilon)2^m}\\
&\leq  \prob{S}{\mathcal{B} \; | \; |\overline{S}| = (1-p)(1-\epsilon)2^m} \;.
\end{align*}
We are left to bound $\prob{}{\mathcal{B} \; | \; |\overline{S}| = (1-p)(1-\epsilon)2^m}$. Note that $(1-p)(1-\epsilon)2^m = c(1-\epsilon) \binom{m}{\leq r}$ and denote this quantity by $s$. By \Cref{recover from erasure equivalent to support of codewords} the probability that we cannot recover from $2^m - s$ random erasures equals
$$\mathrm{Pr}_{S \subseteq \F_2^m, |S|=2^m - s}\left[\exists f \in \polynomials{m}{r} \;\;\; \support{f}\subseteq S\right] \;,$$
where $S \subseteq \F_2^m$ is a random erasure pattern of size exactly $2^m - s$. Calculating we get that
\begin{align*}
\prob{S}{\mathcal{A}_2} &\leq \prob{S}{\mathcal{B} \; | \; |\overline{S}| = (1-p)(1-\epsilon)2^m}\\
&= \mathrm{Pr}_{S}\left[\exists f \in \polynomials{m}{r} \;, \; f \neq 0 \;, \; \support{f}\subseteq S  \; | \; |\overline{S}| = s\right] \\
&= \mathrm{Pr}_{\overline{S}}\left[\exists f \in \polynomials{m}{r} \;, \; f \neq 0 \;, \; \overline{S} \subseteq \overline{\support{f}}  \; | \; |\overline{S}| = s\right] \\
&\leq \sum_{0 \neq f \in \polynomials{m}{r}}\mathrm{Pr}_{\overline{S} \subseteq \F_2^m}\left[\overline{S}\subseteq \overline{\support{f}} \; | \; |\overline{S}| = s \right]\\
&= \sum_{0 \neq f \in \polynomials{m}{r}} \frac{\binom{(1-\weight{f})2^m}{s}}{\binom{2^m}{s}}\\
&= \sum_{0 \neq f \in \polynomials{m}{r}} \frac{(1-\weight{f})2^m \cdots ((1-\weight{f})2^m-s+1)}{2^m\cdot(2^m-1) \cdots (2^m-s+1)}\\
&\leq \sum_{0 \neq f \in \polynomials{m}{r}} (1-\weight{f})^{s} \;.
\end{align*}
Since $s=c(1-\epsilon) \binom{m}{\leq r}$ and $\epsilon = o(1)$ we get the claimed bound.
\end{proof}

\subsection{Reed-Muller code achieves capacity for the BEC}
\label{section: rm achieve capacity bec}
In this section we show that for any $r \leq m/50$ the family of RM codes $\reedmuller{m}{r}$ achieve capacity for random erasures.

\begin{customthm}{\ref{main thm - capacity for bec}}
For any $\gamma \leq 1/50$ the RM code $\reedmuller{m}{\gamma m}$ achieves capacity for random erasures.
\end{customthm}

\begin{proof}
We start by proving the theorem while assuming that $\gamma$ is sufficiently small and then show that $\gamma = 1/50$ suffices. Consider the RM code $\reedmuller{m}{r}$ where $r = \gamma m$ and $\gamma$ is some positive constant to be determined later. We need to show that for any $\delta>0$ it holds that $\lambda_{\mathrm{BEC}}\per{p,\gamma} = o(1)$ where $R$ is the rate of $\reedmuller{m}{\gamma m}$ and $p = 1 - (1+\delta)R$. Applying \Cref{probability of recovering from random erasures union bound} with $c= (1+\delta)$ it suffices to prove that for any $\delta>0$ it holds that,
\begin{equation}
\label{eq: random erasures}
 \sum_{0 \neq f \in \polynomials{m}{r}} (1-\weight{f})^{(1+\delta -o(1)) \binom{m}{\leq r}} \leq \sum_{0 \neq f \in \polynomials{m}{r}} (1-\weight{f})^{(1+\delta/2)\binom{m}{\leq r}} = o(1) \;.
\end{equation}
Let $\delta>0$ and without the loss of generality assume $\delta<1/100$. We shall partition the summands in \cref{eq: random erasures} to three sets:
\begin{itemize}
\item 
Typical:
Polynomials with extremely small bias (including negative bias), i.e. all polynomials $f$ satisfying $\bias{f} \leq \delta/8$.

\item
Relatively small bias:
Polynomials with not too large bias: $\delta/8 \leq \bias{f} \leq \frac{3}{4}$.

\item
Low weight:
Polynomials of weight $\weight{f} \leq \frac{1}{8}$.
\end{itemize}
Next we show that in all three cases the sum in \cref{eq: random erasures} is $o(1)$, which implies the claim. 

\paragraph{Typical case:} In this case, we bound the number of typical polynomials by the number of all degree $r$ polynomials. This is a crude bound (though from \cref{main thm - low bias estimation} we know that this estimate is not too far from the truth) but it suffices for our needs. Since the weight of each typical polynomial is at least $(1-\delta/8)/2$ it follows that
\begin{align*}
\sum_{\bias{f}\leq \delta/8} \per{1-\weight{f}}^{(1+\delta/2)\binom{m}{\leq r}} &\leq 2^{\binom{m}{\leq r}}\cdot \per{1 - \frac{1-\delta/8}{2}}^{(1+\delta/2)\binom{m}{\leq r}}\\
&= 2^{\binom{m}{\leq r}} \cdot \per{\frac{1+\delta/8}{2}}^{(1+\delta/2)\binom{m}{\leq r}}\\
&= \per{(1+\delta/8)^{2/\delta}\cdot \per{\frac{1+\delta/8}{2}}}^{\frac{\delta}{2}\binom{m}{\leq r}}\\
&\leq \per{e^{1/4}\cdot \per{\frac{1+\delta/8}{2}}}^{\frac{\delta}{2}\binom{m}{\leq r}}\\
&\leq e^{-\frac{\delta}{6}\binom{m}{\leq r}} \;.
\end{align*}
Note that in the last inequality we used that for $\delta<1/3$ we have,
$$e^{1/4}\per{\frac{1+\delta/8}{2}} \leq e^{-1/3} \;.$$
This concludes the typical case. 

\paragraph{Low weight case:} Partition the polynomials of weight $\leq \frac{1}{8}$ to dyadic intervals by considering,
$$P_{\ell} = \set{f \in \reedmuller{m}{r} : 2^{-\ell - 1 } \leq \weight{f} \leq 2^{-\ell}}$$
for $\ell = 3,4,\ldots,r$. Every polynomial in $P_{\ell}$ has weight at least $2^{-\ell-1}$ and there are at most $\weightdistribution{m}{r}{2^{-\ell}}$ such polynomials. Using \Cref{main thm - low weight} we get,
$$W_{m,r} (2^{-\ell})\leq \exp_2\per{O(m^4) + 17(c_{\gamma}\ell+d_{\gamma})\gamma^{\ell-1}\binom{m}{\leq r}} \;,$$
where $c_{\gamma} = \frac{1}{1-\gamma}$, $d_{\gamma} = \frac{2-\gamma}{(1-\gamma)^2}$. For sufficiently small $\gamma$ it holds that for any $\ell \geq 3$,
\begin{equation}
\label{eq : low weight condition on gamma}
17(c_{\gamma}\ell+d_{\gamma})\gamma^{\ell-1} \leq \log\per{\frac{1}{1-2^{-\ell-1}}} \;.
\end{equation}
Therefore, 
\begin{align}
\nonumber
&\sum_{\weight{f} \leq 1/8} (1-\weight{f})^{(1+\delta/2)\binom{m}{\leq r}} = \sum_{\ell=3}^{r}\sum_{f \in P_{\ell}} (1-\weight{f})^{(1+\delta/2)\binom{m}{\leq r}}\\
\nonumber
&\leq \sum_{\ell=3}^{r}\weightdistribution{m}{r}{2^{-\ell}}\cdot (1-2^{-\ell-1})^{(1+\delta/2)\binom{m}{\leq r}} \\
\nonumber
&\leq \sum_{\ell=3}^{r}\exp_2\per{O(m^4) + 17(c_{\gamma}\ell+d_{\gamma})\gamma^{\ell-1}\binom{m}{\leq r} -   \log\per{\frac{1}{1-2^{-\ell-1}}}(1+\delta/2)\binom{m}{\leq r}}\\
\label{eq : low weight case 1}
&\leq^{(*)} \sum_{\ell=3}^{r}\exp_2\per{O(m^4) -  \log\per{\frac{1}{1-2^{-\ell-1}}}\delta/2 \binom{m}{\leq r}} \\
\nonumber
&\leq^{(\dagger)} \sum_{\ell=3}^{r}\exp_2\per{O(m^4) -  \delta/\ln(4)\cdot 2^{-\ell-1} \binom{m}{\leq r}} \\
\nonumber
&\leq r\exp\per{-\binom{m}{\leq r}\cdot (\delta/2 - o(1)) 2^{-r -1} } = o(1)\;,
\end{align}
where inequality $(*)$ follows
\cref{eq : low weight condition on gamma}, inequality $(\dagger)$ follows from the fact that, 
$$\log\per{\frac{1}{1-x}} \geq x/\ln(2)$$
 and the last inequality 
 holds since $2^{-r}\binom{m}{\leq r} \geq m 2^{-(\binaryentropy{\gamma}-\gamma -o(1))m} $  and  $\binaryentropy{\gamma} \geq 2\gamma$ for $0\leq \gamma \leq 1/2$,

\paragraph{Relatively small bias case:} Recall that here we handle polynomials with $\delta/8 \leq \bias{f} \leq \frac{3}{4}$. We first deal with the case $\delta/8 \leq \bias{f} \leq \frac{1}{2}$ which leaves out the range $1/8 \leq \weight{f} \leq 1/4$ (that will be analyzed shortly after). The purpose of this distinction is solely to optimize $\gamma$ for which we obtain capacity. Without the loss of generality assume that $\delta$ is some integer power of $1/2$.  Similarly to the low weight case, we are going to consider dyadic intervals. Define,
$$L_{k} = \set{f \in \reedmuller{m}{r} : 2^{-k} \leq \bias{f} \leq 2^{-k+1}}$$
for $k=2,3,\ldots, \log\frac{1}{\delta}+3$. Every polynomial in $L_k$ has weight at least $(1-2^{-k+1})/2$ and there are at most $\weightdistribution{m}{r}{\frac{1-2^{-k}}{2}}$ such polynomials. Apply \Cref{small bias estimation sharper} with $\ell=k$ and $s = k+2$,
$$\weightdistribution{m}{r}{\frac{1-2^{-k}}{2}} \leq \exp_2\per{\per{1-(1-\Tilde{\gamma})^{3k+3} + 17(c_{\gamma} k+c_{\gamma}+d_{\gamma}) \gamma^{k}}\cdot \binom{m}{\leq r} + O(m^4)} \;,$$
where $\Tilde{\gamma} = \gamma\left(1+ \frac{3k+3}{m-3k-3}\right)$. For sufficiently small $\gamma$ and large enoguh $m$ it holds that for $k \geq 2$,
\begin{equation}
\label{eq : low weight condition 2 on gamma}
(1-\Tilde{\gamma})^{3k+3} \geq 17(c_{\gamma} k+c_{\gamma}+d_{\gamma}) \gamma^{k} + (1+\delta/2)\log (1+2^{-k+1})\;.
\end{equation}
Therefore,
\begin{align}
\nonumber
&\sum_{\delta/8 \leq \bias{f} \leq 1/2} (1-\weight{f})^{(1+\delta/2)\binom{m}{\leq r}} = \sum_{k=2}^{\log \frac{1}{\delta}+3} \sum_{f \in L_k} (1-\weight{f})^{(1+\delta/2)\binom{m}{\leq r}}\\
\nonumber
&\leq \sum_{k=2}^{\log \frac{1}{\delta}+3}\weightdistribution{m}{r}{\frac{1-2^{-k}}{2}} \cdot \per{\frac{1+2^{-k+1}}{2}}^{(1+\delta/2)\binom{m}{\leq r}}\\
\nonumber
&= \sum_{k=2}^{\log \frac{1}{\delta}+3}\weightdistribution{m}{r}{\frac{1-2^{-k}}{2}} \cdot 2^{\log \per{\frac{1+2^{-k+1}}{2}} \cdot (1+\delta/2)\binom{m}{\leq r}}\\
\nonumber
&\leq \sum_{k=2}^{\log \frac{1}{\delta}+3}\exp_2\bigg(O(m^4) + \bigg(1-(1-\Tilde{\gamma})^{3k+3} + 17(c_{\gamma} k+c_{\gamma}+d_{\gamma}) \gamma^{k} \\
\nonumber
&\quad\quad\quad+(1+\delta/2)\log \frac{1+2^{-k+1}}{2}\bigg)\binom{m}{\leq r} \bigg)\\
\nonumber
&\leq^{(*)} \sum_{k=2}^{\log \frac{1}{\delta}+3} \exp_2\per{O(m^4) - \delta/2  \binom{m}{\leq r}} =o(1)\;,
\end{align}
where inequality $(*)$ holds due to \cref{eq : low weight condition 2 on gamma}. To complete the ``relatively small bias'' case we need to show that,
$$\sum_{1/8 \leq \weight{f} \leq 1/4} (1-\weight{f})^{(1+\delta/2)\binom{m}{\leq r}} = o(1) \;.$$
Apply \Cref{small bias estimation sharper} with $\ell = 1$ and $s=6$. Thus, 
$$\weightdistribution{m}{r}{1/4} \leq \exp_2\per{O(m^4)+\per{1-(1-\Tilde{\gamma})^{9} + 17(5 c_{\gamma}+d_{\gamma}) \gamma^{4}}\cdot \binom{m}{\leq r}} \;,$$
where $\Tilde{\gamma} = \gamma(1+ \frac{9}{m-9})$, $c_{\gamma} = \frac{1}{1-\gamma}$ and $d_{\gamma} = \frac{2-\gamma}{(1-\gamma)^2}$. For sufficiently small $\gamma$ we may assume that,
\begin{equation}
\label{eq : low weight condition 3 on gamma}
(1-\Tilde{\gamma})^{9} \geq 17(5 c_{\gamma}+d_{\gamma}) \gamma^{4} + \log(7/4) \;.
\end{equation}
Therefore as we did in the low weight cases,
\begin{align*}
&\sum_{1/8 \leq \weight{f} \leq 1/4} (1-\weight{f})^{(1+\delta/2)\binom{m}{\leq r}} \leq \weightdistribution{m}{r}{1/4} \cdot \per{1-\frac{1}{8}}^{(1+\delta)\binom{m}{\leq r}}\\
&\leq \exp_2\per{O(m^4)+\per{1-(1-\Tilde{\gamma})^{9} + 17(5 c_{\gamma}+d_{\gamma}) \gamma^{4} + \log(7/8) (1+\delta/2)}\cdot \binom{m}{\leq r}} \\
&\leq^{(\dagger)} \exp_2\per{O(m^4)- \delta/2 \binom{m}{\leq r}} = o(1) \;,
\end{align*}
where as before, inequality $(\dagger)$ follows from \cref{eq : low weight condition 3 on gamma}.
This finalizes the proof in the ``relatively small bias'' case and so the entire proof is complete for small enough $\gamma$.\\

We next show that $\gamma = 1/50$ suffices for the argument to work. Going over the proof we see that $\gamma$ has to satisfy the constraints in \cref{eq : low weight condition on gamma}, \cref{eq : low weight condition 2 on gamma}, \cref{eq : low weight condition 3 on gamma}. Note that  in \cref{eq : low weight condition 2 on gamma} and \cref{eq : low weight condition 3 on gamma} we have $\Tilde{\gamma}$ as well, but since $\Tilde{\gamma}=\gamma+o(1)$ let us first consider those equations with $\gamma$ alone:
\begin{align*}
\log\per{\frac{1}{1-2^{-\ell-1}}} &\geq 17(c_{\gamma}\ell+d_{\gamma})\gamma^{\ell-1}  &  \ell \geq 3 \;,\\
(1-{\gamma})^{3k+3} &\geq 17(c_{\gamma} k+c_{\gamma}+d_{\gamma}) \gamma^{k} + (1+\delta)\log (1+2^{-k+1}) &    k \geq 2 \;,\\
(1-{\gamma})^{9} &\geq 17(5 c_{\gamma}+d_{\gamma}) \gamma^{4} + \log(7/4)   \;,
\end{align*}
where $c_{\gamma} = \frac{1}{1-\gamma}$, $d_{\gamma} = \frac{2-\gamma}{(1-\gamma)^2}$. It is straightforward to verify that the above inequalities holds for all $\gamma \leq 1/50$. Now, since in both \cref{eq : low weight condition 2 on gamma} and \cref{eq : low weight condition 3 on gamma}  we have that, say, $\Tilde{\gamma} =  \gamma + o(1)$ and we can pick $\delta$ not too small, we get that the original inequalities are satisfied for every $\gamma\leq 1/50$.
\end{proof}

\subsection{Reed-Muller codes of degrees $(1/2-o(1))m$}\label{sec:RM-erasures-1/2}

We now prove that if we relax a bit the requirement that $p=1-(1+o(1))R$ then we can show that RM codes of degrees $(1/2 - \epsilon)m$ can handle a fraction of $1-o(1)$ random erasures. 

\begin{customthm}{\ref{main thm - noise for bec}}
For any $\gamma < 1/2 - \Omega\per{\frac{\sqrt{\log m}}{\sqrt{m}}}$, $\reedmuller{m}{\gamma m}$ can efficiently decode a fraction of $1-o(1)$ random erasures.
\end{customthm}

\begin{proof}
Let $\gamma < 1/2 - \Omega\per{\frac{\sqrt{\log m}}{\sqrt{m}}}$ and set $p= 1 - D_{\gamma} R$ where $R$ is rate of $\reedmuller{m}{\gamma m}$ and $D_{\gamma}$ is a positive constant we shall determine later (which depends on $\gamma$). Applying \Cref{probability of recovering from random erasures union bound} with $D_{\gamma}$ we get,
$$\lambda_{\mathrm{BEC}}(p,\gamma ) \leq \exp\per{-\Theta\per{2^{\frac{\binaryentropy{\gamma}}{2}m}}} + \sum_{0 \neq f \in \polynomials{m}{r}} (1-\weight{f})^{D_{\gamma}(1-o(1))\binom{m}{\leq r}} \;.$$
We proceed just as we did in the proof of \ref{main thm - capacity for bec} by partitioning the summands to three sets: typical, small bias and low weight. In fact, it suffices to break the sum to two sets: polynomials with weight at least $1/4$ and polynomials with weight at most $1/4$. Using the trivial upper bound on the number of polynomials with weight at least $1/4$, which is simply the number of all degree $r$ polynomials, we conclude that,
$$\sum_{\weight{f} \geq 1/4} (1-\weight{f})^{D_{\gamma}(1-o(1))\binom{m}{\leq r}} \leq 2^{\binom{m}{\leq r}} \cdot (3/4)^{D_{\gamma}(1-o(1)) \binom{m}{\leq r}} = o(1).$$
assuming $D_{\gamma} > 3$. To deal with the set of polynomials with weight smaller than $1/4$ we do exactly as in the proof of \Cref{main thm - capacity for bec} and partition it further to dyadic intervals. Going over the analysis reveals that we just need the following inequality to hold (see equations \ref{eq : low weight condition on gamma} and \ref{eq : low weight case 1}),
$$18(c_{\gamma}\ell+d_{\gamma})\gamma^{\ell-1} < \log\per{\frac{1}{1-2^{-\ell-1}}}D_{\gamma}(1-o(1)) \;,$$
for all $\ell \geq 2$. Using that $\log\per{\frac{1}{1-2^{-\ell-1}}} = O(2^{-\ell})$ we conclude that this clearly holds for large enough $D_{\gamma}$. In fact, it is not hard to see that $D_\gamma = O\left(\frac{1}{1-2\gamma}\right)$ suffices.

Observe that when $\gamma = 1/2 - \Omega\per{\frac{\sqrt{\log m}}{\sqrt{m}}}$ we have that 
$$R = \exp_2\per{(h(\gamma)-1)m-\Theta(\log m)} = \exp_2\per{-\Omega(\log m)} = \frac{1}{\poly(m)}$$
and hence $1-D_\gamma \cdot R = 1 - O\per{\frac{\sqrt{m}}{\sqrt{\log m}}}\cdot \frac{1}{\poly(m)} = 1-o(1)$.
\end{proof}

\section{Reed-Muller codes under random errors}\label{sec:BSC}
In this section we study the behavior of RM codes under random errors. In \Cref{section: RM achieve capacity for bsc} we prove \Cref{main thm - capacity for bsc} showing that RM codes achieve capacity for the BSC for degrees at most $m/70$ (with respect to the maximum likelihood decoder). In section \Cref{sec:RM-errors-1/2} we prove \Cref{main thm - noise for bsc} showing that for degrees up to $(1/2-o(1))m$ (for some explicit $o(1)$ function) RM codes can recover from $1/2-o(1)$ random errors (with respect to the maximum likelihood decoder).

Throughout this section we denote $r = \gamma m$ where $0 < \gamma < 1$ is some parameter that can be a constant or a function of $m$. We denote the probability that the family $\reedmuller{m}{\gamma m}$ cannot recover from random erasures with parameter $p = p(m,\gamma)$ with $\lambda_{\mathrm{BSC}}(p,\gamma)$ (See \Cref{recover from errors erasures}).

An important relation that we shall constantly use   is the Taylor expansion of the binary entropy function around $1/2$. Given $p \in (0,1/2)$ write $p = \frac{1-\xi}{2}$ then
\begin{equation}
\label{eq: taylor of binary entropy}
\binaryentropy{p} = 1 - \frac{1}{2\ln(2)}\sum_{k=1}^{\infty} \frac{\xi^{2k}}{k(2k-1)}  \;.    
\end{equation}
Also, denote the probability that the RM code $\reedmuller{m}{\gamma m}$ cannot recover from random errors with parameter $p$ by $\lambda_{\mathrm{BSC}}(p,\gamma)$ (See \Cref{recover from errors erasures}) and note that $\lambda_{\mathrm{BSC}}(p,\gamma) = \prob{z}{\mathrm{ML}(z) \neq y}$ where $z \sim \bsc{p}(y)$ (this quantity is independent of $y$ according to \Cref{equation for bsc} which we shall soon prove).

In the BSC, every bit is flipped with probability $p$ and remains unchanged otherwise. It is convenient to write $z \sim \bsc{p}(y)$ as $z = y + v$ where $v \sim \set{0,1}^{2^m}$ is a random binary vector such that $v_i = 1$ with probability $p$ and $v_i = 0$ with probability $1-p$ independently for every coordinate. We shall refer to $v$ as the error pattern. Suppose we are given $z \sim \bsc{p}(y)$ and use the ML decoder to decode $z$. Every possible decoding $y' \in \reedmuller{m}{r}$ defines an error pattern $v_{y'} = z + y'$ and so $y$ defines a collections of possible errors patterns $\set{v_{y'} : y' \in \reedmuller{m}{r}}$. The ML decoder inspects all possible decodings and decodes $z$ to $y'$ where $v_{y'}$ has minimal weight (in our case this is the most likely outcome of the decoder). Thus, we need to prove that w.h.p $\weight{v_{y}} < \weight{v_{y'}}$ for any $y' \neq y$. \\

We now give an upper bound on $\lambda_{\mathrm{BSC}}(p,\gamma)$ that expresses it in terms of the weight distribution of $\reedmuller{m}{r}$.
\begin{lem}
\label{equation for bsc}
Let $\gamma = \gamma(m) \in (0,1/2)$ be some parameter for any $m \in \N$. Let $r = \gamma m$ and denote $R$ the rate of the RM code $\reedmuller{m}{r}$. For a parameter $c = c(m) \geq 1$ 
 let $p_c$ satisfy $ \binaryentropy{p_c}= 1 - c \cdot R$. Let $\xi$ be such that $p_c = \frac{1-\xi}{2}$. Then for large enough $m$ it holds that 
$$
\lambda_{\mathrm{BSC}}(p,\gamma) \leq o(1)  + \sum_{0,1 \neq f \in \polynomials{m}{r}}  \exp_2\per{-\binom{m}{\leq \gamma m}\per{c\cdot\frac{\weight{f}}{1-\weight{f}}(1-o(1)}} \;.
$$
\end{lem}
The proof is similar in spirit to the proof of \Cref{probability of recovering from random erasures union bound} except that the calculations are a bit more complicated.
\begin{proof}
Since $c$ is fixed throughout the proof we drop the subscript $c$ and denote $p = p_c$. Let $\epsilon = 2^{-(1/2-\binaryentropy{\gamma}/5)m}$.
We first show that 
\begin{equation}\label{eq:eps-xi}
\epsilon = o(\xi)\;.
\end{equation}
From the Taylor expansion in \cref{eq: taylor of binary entropy} we have 
\begin{equation}\label{eq:Cr=xi2}
c R = \frac{\xi^2}{2\ln(2)} + \Theta(\xi^4) = \frac{\xi^2}{2\ln(2)} (1+o(1))
\end{equation} 
so $\xi = \Theta(\sqrt{cR})$. Since $c \geq 1$ we get 
$$\xi = \Omega\per{\sqrt{R}} = \Omega\per{2^{-(1/2-\binaryentropy{\gamma}/2)m -O(\log m) }} = \omega\per{2^{-(1/2-\binaryentropy{\gamma}/5)m}}=\omega\per{\epsilon} \;.$$ 

Let $y \in \reedmuller{m}{r}$ and $z \sim \bsc{p}(y)$. Let $v = z- y$. Then $v$ is a random error pattern such that $v_i = 1$ with probability $p$ and $v_i = 0$ with probability $1-p$, for every coordinate, independently. Typically, an error pattern $v$ has roughly $p$ fraction of errors. 
Consider the following two bad events: 
\begin{itemize}
    \item 
    The error pattern $v$ is not typical, namely $\abs{\weight{v} - p} > \epsilon$. Denote this event by $\mathcal{A}_1$
    
    \item
    The error pattern $v$ is typical, namely $\abs{\weight{v} - p} \leq \epsilon$, but we make an error on it. Denote this event by $\mathcal{A}_2$.
\end{itemize}
By the union bound, 
$$\prob{z}{\mathrm{ML}(z) \neq y} \leq \prob{}{\mathcal{A}_1} + \prob{}{\mathcal{A}_2}\;.$$
We start by bounding the probability that $\mathcal{A}_1$ occurs. By Hoeffding's inequality (\Cref{thm:hoeffding}) we get that,
$$\Pr[\mathcal{A}_1] \leq 2\exp\per{-2\epsilon^2 2^m}2\exp\per{-2 \cdot 2^{2\binaryentropy{\gamma}/5m}}=o(1)\;.$$

Next we bound $\Pr[\mathcal{A}_2]$. We start by conditioning on the number of errors,
$$\Pr[\mathcal{A}_2] = \sum_{\nu = (1-\epsilon)p}^{(1+\epsilon)p} \prob{v}{\weight{v} = \nu} \cdot \prob{v}{\mathrm{ML}(y+v) \neq y | \weight{v} = \nu}\;.$$
Clearly, the function $\prob{v}{\mathrm{ML}(y+v) \neq y | \weight{v} = \nu}$ is non decreasing as a function of $\nu$. I.e., the more errors there are the less likely $y$ is the closest word to $z=y+v$. Hence,
\begin{align*}
\prob{}{\mathcal{A}_2} &= \sum_{\nu = (1-\epsilon)p}^{(1+\epsilon)p} \prob{v}{\weight{v} = \nu} \cdot \prob{v}{\mathrm{ML}(y+v) \neq y | \weight{v} = \nu} \\
&\leq \sum_{\nu = (1-\epsilon)p}^{(1+\epsilon)p} \prob{v}{\weight{v} = \nu} \cdot \prob{v}{\mathrm{ML}(y+v) \neq y | \weight{v} = p(1+\epsilon)}\\
&\leq  \prob{v}{\mathrm{ML}(y+v) \neq y | \weight{v} = p(1+\epsilon)} \;.
\end{align*}
We are left to bound $\prob{v}{\mathrm{ML}(y+v) \neq y | \weight{v} = p(1+\epsilon)}$. Denote $\tilde{p} = p(1+\epsilon)$. Let $\tilde{\xi}$ be such that  $\tilde{p} = \frac{1-\tilde{\xi}}{2}$. Then $\tilde{\xi} = (\xi+\xi\epsilon-\epsilon)$.

The ML decoder will fail to decode correctly if there exists another codeword in the ball at radius $\tilde{p}2^m$ around $z$. I.e. if there is another error pattern $v' \neq v$ of weight at most $\tilde{p}$ such that $v'+z$ is also a codeword. Hence, we say that $v$ is ``bad'' if there exists $v' \neq v$ such that $v'+z$ is a codeword and $\weight{v'} \leq \weight{v} = \tilde{p}$. Equivalently, $v$ is ``bad'' if $v+v'$ is a nonzero codeword and $\weight{v'} \leq \weight{v} = \tilde{p}$. 
We claim that given a codeword $f \in \reedmuller{m}{\gamma m}$ with  weight $\beta = \weight{f}$ we have that,
\begin{equation}
\label{eq: bad pairs bsc}
\abs{\set{(v,v') : \text{$v'+v = f$ and $\weight{v'} \leq \weight{v} = \tilde{p}$}}} \leq 2^{\beta 2^m} \binom{2^m(1-\beta)}{\leq 2^m(\tilde{p} - \beta/2)} \;.
\end{equation}
To see this, observe that $v,v'$ must partition $\support{f}$, and must coincide outside $\support{f}$ (i.e, $v_i \neq v_i'$ if and only if $i \in \support{f}$). There are at most $2^{\beta 2^m}$ possibilities to choose the coordinates of $v$ on $\support{f}$. Also, as $\weight{v'} \leq \weight{v}$ and $v,v'$ must coincide outside $\support{f}$ then $\abs{v_i = 1 : i \in \support{f}} \geq 2^m\cdot\beta/2$ and hence 
$$\abs{\set{v_i = 1 : i \not\in \support{f}}} \leq 2^m(\tilde{p} - \beta/2) \;.$$
Thus there are at most $\binom{2^m(1-\beta)}{\leq 2^m(\tilde{p} - \beta/2)} $ possibilities to choose the coordinates of $v$ on $\overline{\support{f}}$. Since $v$ completely determines $v'$ we deduce \cref{eq: bad pairs bsc}. Next observe that,
\begin{align}
\prob{z}{\mathrm{ML}(z) \neq y | \weight{z-y} = \tilde{p}} &= \frac{\abs{\set{v : \text{$v$ is ``bad''}, \weight{v} = \tilde{p}}}}{\binom{2^m}{\tilde{p}2^m}} \nonumber \\
&\leq \frac{\abs{\set{(v,v') : \text{$v + v'$ is a nonzero codeword and $\weight{v'}\leq \weight{v} = \tilde{p}$ }}}}{\binom{2^m}{\tilde{p}2^m}} \nonumber \\    
&\leq \sum_{0 \neq f \in \reedmuller{m}{r}} \frac{\abs{\set{(v,v') : \text{$v + v' = f$ and $\weight{v'}\leq \weight{v} = \tilde{p}$ }}}}{\binom{2^m}{\tilde{p}2^m}} \;. \label{eq:MLz1}    
\end{align}
Substituting \cref{eq: bad pairs bsc} to \cref{eq:MLz1}  and applying the upper bound on $\weight{f}$ gives
\begin{align}
\label{eq: union bound in bsc lemma}
\prob{z}{\mathrm{ML}(z) \neq y} \leq \sum_{\substack{0 \neq f \in \reedmuller{m}{r} \\ \weight{f} \leq 1 - \Omega(\xi)}} \frac{2^{2^m \weight{f}}\binom{2^m - \absoluteweight{f}}{\leq \per{\tilde{p}2^m - \absoluteweight{f}/2}}}{\binom{2^m}{\tilde{p}2^m}} \;.    
\end{align}
We bound each summand separately. Fix $f \in \reedmuller{m}{r}$ such that $f \neq 0$ and set $\beta = \weight{f}$, $w = \beta 2^m$. Using \cref{amir comb lem} we have,
\begin{align}
&  \frac{2^{\beta 2^m} \binom{2^m(1-\beta)}{\leq 2^m(\tilde{p} - \beta/2)}}{\binom{2^m}{\tilde{p}2^m}} \leq\frac{\exp_2\per{\beta 2^m+2^m(1-\beta)\binaryentropy{\frac{\tilde{p}-\beta/2}{1-\beta}}}}{\exp_2\per{2^m\binaryentropy{\tilde{p}} - O(m)}}\nonumber\\
&\leq \exp_2\per{2^m\per{\beta+(1-\beta)\binaryentropy{\frac{\tilde{p}-\beta/2}{1-\beta}}-\binaryentropy{\tilde{p}}+2^{-m}O(m)}}\nonumber\\
&= \exp_2\per{2^m\per{1-\binaryentropy{\tilde{p}} - (1-\beta)\per{1-\binaryentropy{\frac{\tilde{p}-\beta/2}{1-\beta}}}+2^{-m}O(m)}}\;.
\label{eq:part-mlz}
\end{align}
As $\tilde{p} = \frac{1-\tilde{\xi}}{2}$
we get from \cref{eq: taylor of binary entropy} that
\begin{equation}
1-\binaryentropy{\tilde{p}} = \frac{1}{2\ln(2)}\sum_{k=1}^{\infty} \frac{\tilde{\xi}^{2k}}{k(2k-1)} \;.
\label{eq:1-h1}
\end{equation}
Similarly,
\begin{equation}
1-\binaryentropy{\frac{\tilde{p}-\beta/2}{1-\beta}} =1-\binaryentropy{\frac{1}{2}-\frac{\tilde{\xi}}{2(1-\beta)}} =\frac{1}{2\ln(2)}\sum_{k=1}^{\infty} \frac{\tilde{\xi}^{2k}}{k(2k-1)(1-\beta)^{2k}} \;.
\label{eq:1-h2}
\end{equation}
Substituting equations \eqref{eq:1-h1} and \eqref{eq:1-h2} to \eqref{eq:part-mlz} we thus obtain,
\begin{align}
 \eqref{eq:part-mlz} &\leq \exp_2\bigg(2^m\bigg(\frac{1}{2\ln(2)}\sum_{k=1}^{\infty}\per{ \frac{\tilde{\xi}^{2k}}{k(2k-1)}  -(1-\beta)  \frac{\tilde{\xi}^{2k}}{k(2k-1)(1-\beta)^{2k}} } + 2^{-m}O(m)\bigg)\bigg)\nonumber\\
 &= \exp_2\bigg(2^m\bigg(\frac{1}{2\ln(2)}\sum_{k=1}^{\infty}\per{ \frac{\tilde{\xi}^{2k}}{k(2k-1)}\per{1  -  \frac{1}{(1-\beta)^{2k-1}} }} + 2^{-m}O(m)\bigg)\bigg)\nonumber\\
\nonumber
&\leq \exp_2\bigg(2^m\bigg(  -\frac{\beta \tilde{\xi}^2}{2\ln(2)(1-\beta)} + 2^{-m}O(m)\bigg)\bigg)\\
&=  \exp_2\bigg(2^m\bigg(  -\frac{\beta \tilde{\xi}^2}{2\ln(2)(1-\beta)}(1-o(1)) \bigg)\bigg)\;. \label{eq:mlz-end}
\end{align}
\Cref{eq:Cr=xi2} gives
\begin{equation}\label{eq:xit-R}
\frac{\tilde{\xi}^2}{2\ln 2} = \frac{\per{\xi+\xi\epsilon-\epsilon}^2}{2\ln 2} = \frac{\xi^2(1+\epsilon-\epsilon/\xi)^2}{2\ln 2} = \frac{\xi^2(1-o(1))}{2\ln 2} = cR(1-o(1))\;.
\end{equation}
Thus,
$$\eqref{eq:mlz-end} =  \exp_2\bigg(2^m\bigg(  -\frac{\beta cR}{1-\beta}(1-o(1)) \bigg)\bigg) = \exp_2\bigg(-{m\choose \gamma m}\bigg( c \frac{\beta }{1-\beta}(1-o(1)) \bigg)\bigg)  \;.
$$
Combining this inequality with \cref{eq: union bound in bsc lemma} completes the proof. Note that we sum over all $f \neq 0,1$ rather than all $0 \neq f \in \reedmuller{m}{ \gamma m}$ with weight at most $1-\Omega(\xi)$ simply to ease notation.
\end{proof}

\subsection{Reed-Muller code achieves capacity for the BSC}
\label{section: RM achieve capacity for bsc}
In this section we show that for any $r \leq m/70$ then family of RM codes $\reedmuller{m}{r}$ achieve capacity for random errors.
\begin{customthm}{\ref{main thm - capacity for bsc}}
For any $\gamma \leq 1/70$ the family $\reedmuller{m}{\gamma m}$ achieves capacity for random errors.
\end{customthm}
\begin{proof}
We start by proving the theorem while assuming that $\gamma$ is sufficiently small and then show that $\gamma = 1/70$ suffices. Consider the RM code $\reedmuller{m}{r}$ where $r = \gamma m$ and $\gamma$ is some positive constant to be determined later. We need to show that for every $\delta>0$, if we let $p$ be such that $1-\binaryentropy{p} =  (1+\delta)R$,  where $R$ is the rate of $\reedmuller{m}{\gamma }$, then
 $\lambda_{\mathrm{BSC}}\per{p,\gamma} = o(1)$. Applying \Cref{equation for bsc} for $c=(1+\delta)$ we see that it suffices to prove that,
$$\sum_{0 , 1\neq f \in \polynomials{m}{r}}  \exp_2\per{-\binom{m}{\leq r} \frac{\weight{f}}{1-\weight{f}} (1+\delta-o(1))} = o(1) \;,$$
for every $\delta>0$.
Using \cref{eq: taylor of binary entropy} we get that $1 - \binaryentropy{p} = \frac{\xi^2}{2 \ln(2)}+\Theta(\xi^4)$ and since $1-\binaryentropy{p} = (1+\delta)R$ we get that $\xi^2 = \Theta(R)$. Assuming $m$ is large enough we may settle for,
\begin{equation}
\label{eq: random errors}
\sum_{0 , 1 \neq f \in \polynomials{m}{r}}  \exp_2\per{-\binom{m}{\leq r}\per{\weight{f}\cdot (1+\delta/2)}} = o(1) \;.
\end{equation}
Let $\delta>0$ and without the loss of generality assume $\delta<1/100$. We shall partition the summands in \cref{eq: random errors} to three sets:
\begin{itemize}
\item 
Typical:
Polynomials with extremely small bias (including negative bias), i.e. all polynomials $f$ satisfying $\bias{f} \leq \delta/16$.

\item
Relatively small bias:
Polynomials with not too large bias: $\delta/16 \leq \bias{f} \leq \frac{3}{4}$.

\item
Low weight:
Polynomials of weight $\weight{f} \leq \frac{1}{8}$.
\end{itemize}
Next we show that in all three cases the sum in \cref{eq: random errors} is $o(1)$, which implies the claim. 

\paragraph{Typical case:}
In this case, we bound the number of typical polynomials by the number of all degree $r$ polynomials. This is a crude bound (though from \cref{main thm - low bias estimation} we know that this estimate is not too far from the truth) but it suffices for our needs. 
Since the weight of each typical polynomial is at least $(1-\delta/16)/2$ it follows that
\begin{align*}
&\sum_{\bias{f}\leq \delta/8} \exp_2\per{-\binom{m}{\leq r}\per{\frac{\weight{f}}{1-\weight{f}}\cdot (1+\delta/2)}} \\
&\leq \exp_2\per{-\binom{m}{\leq r}\per{\frac{1/2-\delta/16}{1/2+\delta/16}\cdot (1+\delta/2) - 1}}\\
&= \exp_2\per{-\binom{m}{\leq r}\frac{\delta(4-\delta)}{2(\delta+8)}} \;.
\end{align*}
This concludes the typical case.  

\paragraph{Low weight case:}
Partition the polynomials of weight $\leq \frac{1}{8}$ to dyadic intervals by considering,
$$P_{\ell} = \set{f \in \reedmuller{m}{r} : 2^{-\ell - 1 } \leq \weight{f} \leq 2^{-\ell}}$$
for $\ell = 3, 4, \ldots, r$. Every polynomial in $P_{\ell}$ has weight at least $2^{-\ell-1}$ and there are at most $\weightdistribution{m}{r}{2^{-\ell}}$ such polynomials. Using \Cref{main thm - low weight} we get,
$$W_{m,r} (2^{-\ell})\leq \exp_2\per{O(m^4) + 17(c_{\gamma}\ell+d_{\gamma})\gamma^{\ell-1}\binom{m}{\leq r}} \;,$$
where $c_{\gamma} = \frac{1}{1-\gamma}$, $d_{\gamma} = \frac{2-\gamma}{(1-\gamma)^2}$. For sufficiently small $\gamma$ it holds that for $\ell \geq 3$,
\begin{equation}
\label{eq : bsc low weight condition on gamma}
\frac{2^{-\ell-1}}{1-2^{-\ell-1}} \geq 17(c_{\gamma}\ell+d_{\gamma})\gamma^{\ell-1} \;.
\end{equation}
Therefore, 
\begin{align}
\nonumber
&\sum_{\weight{f} \leq 1/8} \exp_2\per{-\binom{m}{\leq r}\per{\frac{\weight{f}}{1-\weight{f}}\cdot(1+\delta/2)}}\\
\nonumber
&= \sum_{\ell=3}^{r}\sum_{f \in P_{\ell}} \exp_2\per{-\binom{m}{\leq r}\per{\frac{\weight{f}}{1-\weight{f}}\cdot(1+\delta/2)}} \\
\nonumber
&\leq \sum_{\ell=3}^{r} \weightdistribution{m}{r}{2^{-\ell-1}} \cdot \exp_2\per{-\binom{m}{\leq r}\per{\frac{2^{-\ell-1}}{1-2^{-\ell-1}}\cdot(1+\delta/2)}}\\
\label{eq : bsc low weight 3}
&\leq^{(*)} \sum_{\ell=3}^{r} \exp_2\per{-\binom{m}{\leq r}\per{\frac{2^{-\ell-1}}{1-2^{-\ell-1}}\cdot(1+\delta/2)-17(c_{\gamma}\ell+d_{\gamma})\gamma^{\ell-1}}}\\
\nonumber
&\leq^{(\dagger)} \sum_{\ell=3}^{r} \exp_2\per{-\binom{m}{\leq r}\cdot \Omega(\delta 2^{-\ell})}\\
\nonumber
&\leq r\exp_2\per{-\binom{m}{\leq r}\cdot \Omega(\delta 2^{-r})} \\
\nonumber
&= o(1) \;.
\end{align}
Inequality $(*)$ is due to \Cref{main thm - low weight} and \Cref{amir comb lem} and inequality $(\dagger)$ is due to \cref{eq : bsc low weight condition on gamma}. This concludes the low weight case. 

\paragraph{Relatively small bias case:}
Recall that here we handle polynomials with $\frac{\delta}{16} \leq \bias{f} \leq \frac{3}{4}$. We first deal with the case $\delta/16 \leq \bias{f} \leq \frac{1}{2}$ which leaves out the range $1/16 \leq \weight{f} \leq 1/4$ (that will be analyzed shortly after). The purpose of this distinction is solely to optimize $\gamma$ for which we obtain capacity. Without the loss of generality assume that $\delta$ is some integeral power of $1/2$.  Similarly to the low weight case, we are going to consider dyadic intervals. Define,
$$L_{k} = \set{f \in \reedmuller{m}{r} : 2^{-k} \leq \bias{f} \leq 2^{-k+1}}$$
for $k=2,3,\ldots, \log\frac{16}{\delta}$. Every polynomial in $L_k$ has weight at least $(1-2^{-k+1})/2$ and there are at most $\weightdistribution{m}{r}{\frac{1-2^{-k}}{2}}$ such polynomials. From \Cref{small bias estimation sharper}, with $\ell=k$ and $s = k+2$, we get that for $\Tilde{\gamma} = \gamma\per{1+\frac{3k+3}{m-3k-3}}$ it holds that
$$\weightdistribution{m}{r}{\frac{1-2^{-k}}{2}} \leq \exp_2\per{\per{1-(1-\Tilde{\gamma})^{3k+3} + 17(c_{\gamma} k+c_{\gamma}+d_{\gamma}) \gamma^{k}}\cdot \binom{m}{\leq r}+O(m^4)} \;.$$
For sufficiently small $\gamma$ and large enoguh $m$ it holds that for $2\leq k \leq \log(16/\delta)$,
\begin{equation}
\label{eq : bsc low weight condition 2 on gamma}
(1-\Tilde{\gamma})^{3k+3} \geq 17(c_{\gamma} k+c_{\gamma}+d_{\gamma}) \gamma^{k} + \frac{2^{-k+2}}{1+2^{-k+1}} \;.
\end{equation}
Therefore,
\begin{align}
\nonumber
&\sum_{\delta/16 \leq \bias{f}\leq 1/2} \exp_2\per{-\binom{m}{\leq r}\per{\frac{\weight{f}}{1-\weight{f}}\cdot(1+\delta/2)}} \\
\nonumber
&= \sum_{k=2}^{\log \frac{16}{\delta}}\sum_{f\in L_k} \exp_2 \per{ -\binom{m}{\leq r} \per{ \frac{ \weight{f}}{ 1-\weight{f}}\cdot(1+\delta/2)}} \\
\nonumber
&\leq \sum_{k=2}^{\log \frac{16}{\delta}} \weightdistribution{m}{r}{\frac{1-2^{-k}}{2}} \cdot\exp_2\per{-\binom{m}{\leq r}\cdot \frac{1/2-2^{-k}}{1/2+2^{-k}}\cdot(1+\delta/2)} \\
\nonumber
&\leq \sum_{k=2}^{\log \frac{16}{\delta}}\exp_2\bigg(-\binom{m}{\leq r}\bigg(-1+(1-\Tilde{\gamma})^{3k+3} - 17(c_{\gamma}k+c_{\gamma}+d_{\gamma})\gamma^{k} \\
\nonumber
&\quad \quad \quad \quad \quad \quad\quad \quad \quad \quad \quad \quad\quad \quad \quad   + \per{1 - \frac{2^{-k+2}}{1+2^{-k+1}}} \cdot (1+\delta/2)\bigg)+O(m^4)\bigg)\\
\label{eq: bsc low bias case 1}
&\leq \sum_{k=2}^{\log \frac{16}{\delta}} \exp_2\bigg(-\binom{m}{\leq r}\bigg((1-\Tilde{\gamma})^{3k+3} - 17(c_{\gamma}k+c_{\gamma}+d_{\gamma})\gamma^{k} - \frac{2^{-k+2}}{1+2^{-k+1}}\cdot(1+\delta/2)\\
\nonumber
&\quad \quad \quad \quad \quad \quad\quad \quad \quad \quad \quad \quad\quad \quad \quad  +\delta/2 - o(1)\bigg)\bigg) \\
\label{eq: bsc low bias case 2}
&\leq \sum_{k=2}^{\log \frac{16}{\delta}} \exp_2\per{-\binom{m}{\leq r}\delta/4}\\
\nonumber
&\leq \per{\log\frac{1}{\delta}+3}\cdot \exp_2\per{-\binom{m}{\leq r}\delta/4}\\
\nonumber
&= o(1)\;.
\end{align}
Inequality (\ref{eq: bsc low bias case 1}) holds for sufficiently large $m$, and inequality (\ref{eq: bsc low bias case 2}) holds due to \cref{eq : bsc low weight condition 2 on gamma} and sufficiently large $m$. To complete the ``relatively small bias'' case we need to handle the case $1/8 \leq \weight{f} \leq 1/4$. I.e., we have to show that
$$\sum_{1/8 \leq \weight{f} \leq 1/4} \exp_2 \per{ -\binom{m}{\leq r} \per{ \frac{ \weight{f}}{ 1-\weight{f}}\cdot(1+\delta/2)}} = o(1) \;.$$
Apply \Cref{small bias estimation sharper} with $\ell = 1$ and $s=6$. Thus, 
$$\weightdistribution{m}{r}{1/4} \leq \exp_2\per{\per{1-(1-\Tilde{\gamma})^{9} + 17(5 c_{\gamma}+d_{\gamma}) \gamma^{4}}\cdot \binom{m}{\leq r}+O(m^4)} \;,$$
where $\Tilde{\gamma} = \gamma\per{1 + O(1/m)}$, $c_{\gamma} = \frac{1}{1-\gamma}$ and $d_{\gamma} = \frac{2-\gamma}{(1-\gamma)^2}$. For sufficiently small $\gamma$ we may assume that,
\begin{equation}
\label{eq : bsc low weight condition 3 on gamma}
(1-\Tilde{\gamma})^{9} \geq 17(5 c_{\gamma}+d_{\gamma}) \gamma^{4} + 6/7 \;.
\end{equation}
Therefore as we did in the low weight cases,
\begin{align*}
&\sum_{1/8 \leq \weight{f} \leq 1/4} \exp_2 \per{ -\binom{m}{\leq r} \per{ \frac{ \weight{f}}{ 1-\weight{f}}\cdot(1+\delta/2)}} \\
&\leq \weightdistribution{m}{r}{1/4}\cdot\exp_2\per{-\binom{m}{\leq r}\per{\frac{1/8}{1-1/8}}(1+\delta/2)}\\
&\leq \exp_2\per{-\binom{m}{\leq r}\per{\frac{1}{7}(1+\delta/2)-1+(1-\Tilde{\gamma})^{9} - 17(5 c_{\gamma}+d_{\gamma}) \gamma^{4}}-o(1)}\\
&\leq^{(\dagger)} \exp_2\per{-\binom{m}{\leq r}\frac{\delta}{14}} = o(1) \;,
\end{align*}
where as before, inequality $(\dagger)$ follows from \cref{eq : bsc low weight condition 3 on gamma}.
This finalizes the proof in the ``relatively small bias'' case and so the entire proof is complete for small enough $\gamma$.\\

We next show that $\gamma = 1/70$ suffices for the argument to work. Going over the proof we see that $\gamma$ has to satisfy the constraints in \cref{eq : bsc low weight condition on gamma}, \cref{eq : bsc low weight condition 2 on gamma}, \cref{eq : bsc low weight condition 3 on gamma}. Similar to the argument in the proof of \Cref{main thm - capacity for bec} we note that equations (\ref{eq : bsc low weight condition 2 on gamma}), (\ref{eq : bsc low weight condition 3 on gamma}) involve both ${\gamma}$ and $\tilde\gamma$ but as $\tilde{\gamma}=\gamma(1+o(1))$ we may consider these equations with $\gamma$ in place of $\tilde{\gamma}$:
\begin{align*}
\frac{2^{-\ell-1}}{1-2^{-\ell-1}} &\geq 17(c_{\gamma}\ell+d_{\gamma})\gamma^{\ell-1}  &  \ell \geq 3 \;,\\
(1-{\gamma})^{3k+3} &\geq 17(c_{\gamma} k+c_{\gamma}+d_{\gamma}) \gamma^{k} + \frac{2^{-k+2}}{1+2^{-k+1}} &    k \geq 2 \;,\\
(1-{\gamma})^{9} &\geq 17(5 c_{\gamma}+d_{\gamma}) \gamma^{4} + 6/7   \;,
\end{align*}
where $c_{\gamma} = \frac{1}{1-\gamma}$, $d_{\gamma} = \frac{2-\gamma}{(1-\gamma)^2}$. It is straightforward to verify that the above inequalities hold for all $\gamma \leq 1/70$, with some constant gap. Hence, for a sufficiently large $m$ and small enough $\delta$ the inequalities also hold for $\tilde{\gamma} = \gamma(1+o(1))$.\end{proof}

\subsection{Reed-Muller codes of degrees $(1/2-o(1))m$}
\label{sec:RM-errors-1/2}
We now prove that if we relax a bit the requirement that $p = 1-(1+o(1))R$ then we can show that RM codes of degrees $(1/2-\epsilon)m$ can handle a fraction of $1/2-o(1)$ errors (using the maximum likelihood decoder). 
\begin{customthm}{\ref{main thm - noise for bsc}}
Let $m \in \N$ and $\gamma(m) < 1/2 - \Omega\per{\frac{\sqrt{m}}{\sqrt{\log m}}}$ a positive parameter (which may depend on $m$). Let $r = \gamma m$. Then, the $\reedmuller{m}{r}$ can decode from a fraction of $1/2-o(1)$ random errors.
\end{customthm}
\begin{proof}
Let $B_{\gamma}(m)$ be a positive parameter  to be determined later (it will depend on $\gamma(m)$). For brevity we write $B_{\gamma}$ and not $B_{\gamma}(m)$.
Let $p$ be such that $1 - \binaryentropy{p} =B_{\gamma} R$ where $R$ is the rate of $\reedmuller{m}{\gamma m}$.  Applying \Cref{equation for bsc} with $B_{\gamma}$ we get,
$$\lambda_{\mathrm{BSC}}(p,\gamma) \leq o(1) + \sum_{0,1 \neq f \in \polynomials{m}{r}}  \exp_2\per{-\binom{m}{\leq r}\per{\frac{\weight{f}}{1-\weight{f}}\cdot B_{\gamma}(1-o(1))}} \;.$$ 
We proceed just as we did in the proof of \Cref{main thm - capacity for bec} and partition the summands 
two sets: polynomials with weight at least $1/4$ and polynomials with weight at most $1/4$. Using the trivial upper bound on the number of polynomials with weight at least $1/4$, which is simply the number of all degree $\gamma m$ polynomials, we conclude that,
$$\sum_{\weight{f} \geq 1/4} \exp_2\per{-\binom{m}{\leq r}\per{\frac{\weight{f}\cdot B_{\gamma}(1-o(1))}{1-\weight{f}} }} \leq 2^{\binom{m}{\leq r}} \cdot \exp_2\per{-\frac{B_{\gamma}(1-o(1))}{3}\binom{m}{\leq r}}  = o(1),$$
which holds assuming $B_{\gamma}>3$. To deal with the set of polynomials with weight smaller than $1/4$ we do exactly as in the proof of \Cref{main thm - capacity for bsc} and partition it further to dyadic intervals. Going over the analysis reveals that we need the following inequality to hold (see equations (\ref{eq : bsc low weight 3}) , (\ref{eq : bsc low weight condition on gamma})),
$$\frac{2^{-\ell-1}}{1-2^{-\ell-1}}\cdot B_{\gamma}(1-o(1)) > 18(c_{\gamma}\ell+d_{\gamma})\gamma^{\ell-1} \;,$$
for all $\ell \geq 2$. This clearly holds for large enough $B_{\gamma}$ (depending on $\gamma$) and in fact it is not hard to see that $B_{\gamma} = O\per{\frac{1}{1-2\gamma}}$ suffices.

Observe that when $\gamma = 1/2 - \Omega\per{\frac{\sqrt{\log m}}{\sqrt{m}}}$ we have that 
$$R = \exp_2\per{(h(\gamma)-1)m-O(\log m)} = \exp_2\per{-\Omega(\log m)} = \frac{1}{\poly(m)}$$
and hence $1-B_\gamma \cdot R = 1 - O\per{\frac{\sqrt{m}}{\sqrt{\log m}}}\cdot \frac{1}{\poly(m)} = 1-o(1)$.
\end{proof}

\section{Discussion}
\label{sec:discuss}
In this work we proved new upper bounds on the weight distribution of Reed-Muller codes. As a result of these bounds we were able to prove that RM codes achieve capacity for the BEC and the BSC for degrees which are linear in the number of variables. While this significantly improved the range of degrees for which RM codes achieve capacity it still leaves several intriguing questions open.

For the BEC we know that RM codes achieve capacity for degrees up to $m/50$ (this work), degrees around $m/2$ (\cite{kumar2015reed}) and very high degrees (\cite{abbe2015reed}). See \autoref{table} on page \pageref{table}. This leaves a wide range of degrees open. It seems unlikely that there would be a gap in the degrees for which RM codes achieve capacity. One possible approach for closing the current gap may be to find a unified way of proving the results in the three different regimes. Currently the proofs are very different from each other. Another possible approach is by strengthening \Cref{main thm - noise for bec}. Recall that this theorem shows that RM codes up to degree nearly $m/2$ can recover from a fraction of $1-o(1)$ random erasures. This is ``almost'' like showing capacity for all such degrees, but not quite so as the $o(1)$ term is not the correct one. Can this result be pushed to obtaining capacity for such high degrees? 

In the case of the BSC the situation is even worse. This work provides the best known results on the range of degrees for which RM codes achieve capacity for the BSC. As in the case of BEC,  \Cref{main thm - noise for bsc} can be seen as a strong evidence that RM codes achieve capacity for degrees up to $m/2 - \Omega\per{\frac{\sqrt{\log m}}{\sqrt m}}$. Unlike in the case of the BEC we do not have any similar result for the BSC for constant rate. Can a result in the spirit of   \Cref{main thm - noise for bsc} be proved, showing that RM codes of degree, say, $m/2$ can tolerate a constant fraction of random errors? It does not necessarily have to be the ``correct'' fraction for getting capacity but rather any constant fraction.

Another interesting question is improving the weight distribution results. \Cref{thm:sam} of Samorodnitsky \cite{DBLP:journals/corr/abs-1809-09696} is proved via a very different approach than ours. Is it possible to extend it and improve our weight distribution results for all degrees?  \Cref{thm:sam} also falls short of providing meaningful information when the bias is small. Can it be improved to speak about this regime of parameters as well?

\newpage

\bibliographystyle{alphaurl}
\bibliography{bibliography}

\newcommand{\etalchar}[1]{$^{#1}$}
\begin{thebibliography}{KKM{\etalchar{+}}17}

\bibitem[AKM{\etalchar{+}}09]{arikan2009performance}
Erdal Ar{\i}kan, Haesik Kim, Garik Markarian, U~Ozgur, and Efecan Poyraz.
\newblock Performance of short polar codes under ml decoding.
\newblock {\em Proc. ICT MobileSummit,(Santander, Spain)}, pages 10--12, 2009.

\bibitem[Ari09]{arikan2009channel}
Erdal Arikan.
\newblock Channel polarization: A method for constructing capacity-achieving
  codes for symmetric binary-input memoryless channels.
\newblock {\em IEEE Transactions on Information Theory}, 55(7):3051--3073,
  2009.

\bibitem[ASW15]{abbe2015reed}
Emmanuel Abbe, Amir Shpilka, and Avi Wigderson.
\newblock Reed--muller codes for random erasures and errors.
\newblock {\em IEEE Transactions on Information Theory}, 61(10):5229--5252,
  2015.

\bibitem[BEHL12]{ben2012random}
Ido Ben-Eliezer, Rani Hod, and Shachar Lovett.
\newblock Random low-degree polynomials are hard to approximate.
\newblock {\em computational complexity}, 21(1):63--81, 2012.

\bibitem[BF90]{beaver1990hiding}
Donald Beaver and Joan Feigenbaum.
\newblock Hiding instances in multioracle queries.
\newblock In {\em Annual Symposium on Theoretical Aspects of Computer Science},
  pages 37--48. Springer, 1990.

\bibitem[BGY18]{beame2018bias}
Paul Beame, Shayan~Oveis Gharan, and Xin Yang.
\newblock On the bias of reed-muller codes over odd prime fields.
\newblock {\em arXiv preprint arXiv:1806.06973}, 2018.

\bibitem[BV10]{bogdanov2010pseudorandom}
Andrej Bogdanov and Emanuele Viola.
\newblock Pseudorandom bits for polynomials.
\newblock {\em SIAM Journal on Computing}, 39(6):2464--2486, 2010.

\bibitem[CGKS95]{chor1995private}
Benny Chor, Oded Goldreich, Eyal Kushilevitz, and Madhu Sudan.
\newblock Private information retrieval.
\newblock In {\em Foundations of Computer Science, 1995. Proceedings., 36th
  Annual Symposium on}, pages 41--50. IEEE, 1995.

\bibitem[Eli55]{elias1955coding}
Peter Elias.
\newblock Coding for noisy channels.
\newblock {\em IRE Conv. Rec.}, 3:37--46, 1955.

\bibitem[Gas04]{gasarch2004survey}
William Gasarch.
\newblock A survey on private information retrieval.
\newblock {\em The Bulletin of the EATCS}, 82(72-107):1, 2004.

\bibitem[GRS12]{guruswami2012essential}
Venkatesan Guruswami, Atri Rudra, and Madhu Sudan.
\newblock Essential coding theory.
\newblock {\em Draft available at http://www. cse. buffalo. edu/~
  atri/courses/coding-theory/book}, 2012.

\bibitem[Ham50]{hamming1950error}
Richard~W Hamming.
\newblock Error detecting and error correcting codes.
\newblock {\em Bell System technical journal}, 29(2):147--160, 1950.

\bibitem[Hoe63]{hoeffding1963probability}
Wassily Hoeffding.
\newblock Probability inequalities for sums of bounded random variables.
\newblock {\em Journal of the American statistical association},
  58(301):13--30, 1963.

\bibitem[KKM{\etalchar{+}}17]{kumar2015reed}
Shrinivas Kudekar, Santhosh Kumar, Marco Mondelli, Henry~D Pfister, Eren
  {\c{S}}a{\c{s}}oǧlu, and R{\"u}diger~L Urbanke.
\newblock Reed--muller codes achieve capacity on erasure channels.
\newblock {\em IEEE Transactions on Information Theory}, 63(7):4298--4316,
  2017.

\bibitem[KLP12]{kaufman2012weight}
Tali Kaufman, Shachar Lovett, and Ely Porat.
\newblock Weight distribution and list-decoding size of reed--muller codes.
\newblock {\em IEEE Transactions on Information Theory}, 58(5):2689--2696,
  2012.

\bibitem[KP18]{DBLP:conf/soda/KoppartyP18}
Swastik Kopparty and Aditya Potukuchi.
\newblock Syndrome decoding of reed-muller codes and tensor decomposition over
  finite fields.
\newblock In Artur Czumaj, editor, {\em Proceedings of the Twenty-Ninth Annual
  {ACM-SIAM} Symposium on Discrete Algorithms, {SODA} 2018, New Orleans, LA,
  USA, January 7-10, 2018}, pages 680--691. {SIAM}, 2018.
\newblock URL: \url{https://doi.org/10.1137/1.9781611975031.44}, \href
  {http://dx.doi.org/10.1137/1.9781611975031.44}
  {\path{doi:10.1137/1.9781611975031.44}}.

\bibitem[KT70]{kasami1970weight}
Tadao Kasami and Nobuki Tokura.
\newblock On the weight structure of reed-muller codes.
\newblock {\em IEEE Transactions on Information Theory}, 16(6):752--759, 1970.

\bibitem[KTA76]{kasami1976weight}
Tadao Kasami, Nobuki Tokura, and Saburo Azumi.
\newblock On the weight enumeration of weights less than 2.5 d of reed-muller
  codes.
\newblock {\em Information and control}, 30(4):380--395, 1976.

\bibitem[McD89]{mcdiarmid1989method}
Colin McDiarmid.
\newblock On the method of bounded differences.
\newblock {\em Surveys in combinatorics}, 141(1):148--188, 1989.

\bibitem[MU05]{MU05-book}
Michael Mitzenmacher and Eli Upfal.
\newblock {\em Probability and Computing -- Randomized Algorithms and
  Probabilistic Analysis}.
\newblock Cambridge University Press, 2005.

\bibitem[Mul54]{muller1954application}
David~E Muller.
\newblock Application of boolean algebra to switching circuit design and to
  error detection.
\newblock {\em Transactions of the IRE Professional Group on Electronic
  Computers}, (3):6--12, 1954.

\bibitem[Pol94]{poltyrev1994bounds}
Gregory Poltyrev.
\newblock Bounds on the decoding error probability of binary linear codes via
  their spectra.
\newblock {\em IEEE Transactions on Information Theory}, 40(4):1284--1292,
  1994.

\bibitem[Ree54]{reed1953class}
I.~Reed.
\newblock A class of multiple-error-correcting codes and the decoding scheme.
\newblock {\em Transactions of the IRE Professional Group on Information
  Theory}, 4(4):38--49, 1954.
\newblock \href {http://dx.doi.org/10.1109/TIT.1954.1057465}
  {\path{doi:10.1109/TIT.1954.1057465}}.

\bibitem[Rob55]{robbins1955remark}
Herbert Robbins.
\newblock A remark on stirling's formula.
\newblock {\em The American mathematical monthly}, 62(1):26--29, 1955.

\bibitem[Sam]{DBLP:journals/corr/abs-1809-09696}
Alex Samorodnitsky.
\newblock An upper bound on {\({l}_{\mbox{q}}\) norms of noisy functions},
  journal = {CoRR}, volume = {abs/1809.09696}, year = {2018}, url =
  {http://arxiv.org/abs/1809.09696}, archiveprefix = {arXiv}, eprint =
  {1809.09696}, timestamp = {Fri, 05 Oct 2018 11:34:52 +0200}, biburl =
  {https://dblp.org/rec/bib/journals/corr/abs-1809-09696}, bibsource = {dblp
  computer science bibliography, https://dblp.org}.

\bibitem[Sha48]{shannon2001mathematical}
Claude~Elwood Shannon.
\newblock A mathematical theory of communication.
\newblock {\em Bell System Technical Journal}, 27(3):379--423, 1948.
\newblock URL:
  \url{https://onlinelibrary.wiley.com/doi/abs/10.1002/j.1538-7305.1948.tb01338.x},
  \href
  {http://arxiv.org/abs/https://onlinelibrary.wiley.com/doi/pdf/10.1002/j.1538-7305.1948.tb01338.x}
  {\path{arXiv:https://onlinelibrary.wiley.com/doi/pdf/10.1002/j.1538-7305.1948.tb01338.x}},
  \href {http://dx.doi.org/10.1002/j.1538-7305.1948.tb01338.x}
  {\path{doi:10.1002/j.1538-7305.1948.tb01338.x}}.

\bibitem[Sha79]{shamir1979share}
Adi Shamir.
\newblock How to share a secret.
\newblock {\em Communications of the ACM}, 22(11):612--613, 1979.

\bibitem[SSV17]{DBLP:journals/tit/SaptharishiSV17}
Ramprasad Saptharishi, Amir Shpilka, and Ben~Lee Volk.
\newblock Efficiently decoding reed-muller codes from random errors.
\newblock {\em {IEEE} Trans. Information Theory}, 63(4):1954--1960, 2017.
\newblock URL: \url{https://doi.org/10.1109/TIT.2017.2671410}, \href
  {http://dx.doi.org/10.1109/TIT.2017.2671410}
  {\path{doi:10.1109/TIT.2017.2671410}}.

\end{thebibliography}

\appendix
\section{Combinatorial lemmas}
\label{Combinatorial Lemmas section}
\begin{lem}
\label{(Simple Combinatorial Bound I)}
Let $\ell \leq r \leq  m \in \N$. Let $0 < \gamma < 1$ be such that $r = \gamma m$. Then,
$$\binom{m-\ell}{\leq r-\ell} \leq \gamma^\ell \binom{m}{\leq r} \;.$$
\end{lem}
\begin{proof}
Note that for any $\ell \leq s\leq m$ we have,
$$\frac{\binom{m-\ell}{s-\ell}}{\binom{m}{s}} = \prod_{j=0}^{\ell-1} \frac{s-j}{m-j} \leq \per{\frac{s}{m}}^\ell \;.$$
It follows that
\begin{align*}
\binom{m-\ell}{\leq r-\ell} &= \binom{m-\ell}{0}+\binom{m-\ell}{1}+\binom{m-\ell}{2} + \ldots + \binom{m-\ell}{r-\ell} \\
&\leq \per{\per{\frac{\ell}{m}}^\ell \cdot \binom{m}{\ell}+ \per{\frac{\ell+1}{m}}^\ell \cdot \binom{m}{\ell+1}+\binom{m}{\ell+2} + \ldots + \per{\frac{r}{m}}^\ell \cdot \binom{m}{r}}\\
&\leq \per{\frac{r}{m}}^\ell \per{\binom{m}{\ell}+\binom{m}{\ell+1}+\binom{m}{\ell+2} + \ldots + \binom{m}{r}}\\
&\leq \per{\frac{r}{m}}^\ell \per{\binom{m}{0}+\binom{m}{1}+\binom{m}{2} + \ldots + \binom{m}{r}}\\
&= \per{\frac{r}{m}}^\ell \binom{m}{\leq r} \;.
\end{align*}
\end{proof}

We next proof \Cref{(Simple Combinatorial Bound III)}. To ease the reading we repeat its statement.
\begin{customlem}{\ref{(Simple Combinatorial Bound III)}}
Let $t,r\leq m \in \N$. Denote $r = \gamma m$. Then, for  $\Tilde{\gamma} = \gamma\left(1+\frac{t}{m-t}\right)$ it holds that
$$\binom{m-t}{\leq r} \geq \per{1-\Tilde{\gamma}}^t\binom{m}{\leq r} \;.$$
\end{customlem}
\begin{proof}
First note that for $k\leq r$,
\begin{align*}
\binom{m-t}{k} = \binom{m}{k}\cdot\prod_{j=0}^{t-1} \frac{m-k-j}{m-j} &\geq \binom{m}{k} \cdot\per{1-\frac{k}{m-t+1}}^t \\
& \geq \binom{m}{k} \cdot\per{1-\frac{r}{m-t}}^t \\
& = \binom{m}{k} \cdot\left(1- \gamma\left(1+\frac{t}{m-t}\right)\right)^t = \binom{m}{k} \cdot\per{1-\Tilde{\gamma}}^t\;.
\end{align*}
It follows that
\begin{equation}
\binom{m-t}{\leq r} = \sum_{ k = 0}^{r} \binom{m-t}{k} \geq \per{1-\Tilde{\gamma}}^t\cdot\sum_{k = 0}^{r} \binom{m}{k}  = \per{1-\Tilde{\gamma}}^t\cdot \binom{m}{\leq r}  \;. 
\label{eq: bounds on binomial sums}
\end{equation}
\end{proof}

\section{Proof of main lemmas of \cite{kaufman2012weight} }
\label{constants in klp section}
We briefly sketch the analysis of \Cref{KPL1} and \Cref{KPL2} as in the original paper \cite{kaufman2012weight} and try to unfold the constants hiding within the asymptotic notation. 

To ease the reading we repeat the statements of the lemmas.

\begin{customlem}{\ref{KPL1}}
Let $f : \F_2^m \rightarrow \F_2$ be a function such that $\weight{f}\leq 2^{-k}$ for $k \geq 2$ and let $\delta > 0$. Then, there exist directions $Y_1,\ldots, Y_t \in (\F_2^m)^{k-1}$ such that
$$\prob{x}{f(x) \neq \text{Maj}\left(\derivative{Y_1}{f}(x),\ldots,\derivative{Y_t}{f}(x)\right)}\leq \delta \;,$$
where $t = \lceil 17\log(1/\delta) \rceil$.
\end{customlem}

The starting point of \Cref{KPL1} is the following simple proposition.
\begin{prop}\label{prop:bias}
Let $f : \F_2^n \rightarrow \F_2^n$ be a function with $\bias{f} \neq 0$. Then,
$$(-1)^{f(x)} = \frac{1}{\bias{f}}\E_{y \in \F_2^n}\left[(-1)^{\derivative{y}{f}(x)}\right] \;.$$
\end{prop}
This suggests that by sampling $(-1)^{\derivative{y}{f}(x)}$ we can evaluate $f(x)$ with high probability. 
Essentially, the proof of \Cref{KPL1} is by generalizing this lemma for high order derivatives while showing (via suitable concentration bound) that a similar sampling works in this case. 

\begin{proof}[Proof of \Cref{KPL1}]
The following proposition is just a repeated application of \Cref{prop:bias}.
\begin{prop}
Let $f : \F_2^n \rightarrow \F_2^n$ be a function with $\bias{f} \leq 2^{-k}$. Then,
\begin{equation}
(-1)^{f(x)} = \E_{Y \in (\F_2^{m})^{k-1}}\left[\alpha_Y \cdot (-1)^{\derivative{Y}{f}(x)}\right] \;,    \label{eq : samples of low weight}
\end{equation}
where, 
$$\alpha_Y = \frac{1}{\bias{f}\cdot \bias{\derivative{y_1}{f}}\cdots \bias{\derivative{y_{k-2}}{\cdots \derivative{y_{1}}{f}}}} \;.$$
\end{prop}

This implies a simple approximation scheme for $f$ -- sample from the distribution $\alpha_Y \cdot (-1)^{\derivative{Y}{f}(x)}$ independently and take the sign of the average. Without the loss of generality assume $f(x) = 1$. Then, for independent random variables $X_i = \alpha_{Y_i} \cdot (-1)^{\derivative{Y_i}{f}(x)}$, sampled independently from $\{\alpha_Y \cdot (-1)^{\derivative{Y}{f}(x)}\}$ we get
\begin{equation}\label{eq:sample}
\prob{}{f(x) \neq \text{Sgn}\per{\frac{1}{t} \sum_{i=1}^{t}X_i}} = \prob{}{{\frac{1}{t}\sum_{i=1}^{t} X_i - (-1)^{f(x)}} \geq 1} \;.
\end{equation}
To analyze this we shall use Hoeffding's Inequality \cite{hoeffding1963probability}.
\begin{thm}
[Hoeffding's Inequality]
Let $X_1,\ldots,X_t$ independent random variables where each $X_i$ is supported on the interval $[a_i,b_i]$. Then,
$$\prob{}{\frac{1}{t}\sum_{i=1}^{t} X_i - \mu \geq \epsilon} \leq \exp\per{\frac{2\epsilon^2 t^2}{\sum_{i=1}^{t}(b_i-a_i)^2}} \;,$$
with $\mu = \E\left[\frac{1}{t}\sum_{i=1}^{t} X_i\right]$.
\end{thm}

Let $M= \max_Y \alpha_Y$. Then, the  random variables $X_i$ in \Cref{eq:sample} are supported on the interval $[-M,M]$.
From Hoeffding's Inequality we get that
$$\prob{}{f(x) \neq \text{Sgn}\per{\frac{1}{t} \sum_{i=1}^{t}X_i}} = \prob{}{{\frac{1}{t}\sum_{i=1}^{t} X_i - (-1)^{f(x)}} \geq 1} \leq \exp\per{-\frac{t}{2M^2}} \;,$$
We thus see that for $$t = \left\lceil {2\ln(1/\delta)}{M^2} \right\rceil  \;,$$
it holds that
$$\prob{x}{f(x) \neq \text{Sgn}\per{\frac{1}{t} \sum_{i=1}^{t}\alpha_{Y_i}\derivative{Y_i}{f}(x)}}\leq \delta \;,$$
where $\alpha_{Y_i} \cdot (-1)^{\derivative{Y_i}{f}(x)}$ are independent random samples from $\{\alpha_Y \cdot (-1)^{\derivative{Y}{f}(x)}\}$. It remains to upper bound $M= \max_Y \alpha_Y$.

\begin{lem}
Let $Y \in (\F_2^{m})^{k-1}$ then,
$$\alpha_Y \leq \frac{1}{\prod_{j=1}^{k}(1-2^{-j})} \leq 3.5 \;.$$
\end{lem}
\begin{proof}
The left inequality follows since derivative may only double the weight hence if $\weight{f} \leq 2^{-k}$ then $\bias{\derivative{Y}{f}} \geq 1-2^{1+t-k}$ for any direction $Y$ of order $t$. For the second inequality, note that it suffices to bound $\prod_{j=1}^{\infty}\frac{1}{1-2^{-j}}$. One can easily verify that $\frac{1}{\prod_{j=1}^{100}(1-2^{-j})} \leq 3.47$. Using that for all $x < 1/2$ we have $\frac{1}{1-x} \leq e^{2x}$ we have,
$$\frac{1}{\prod_{j=101}^{\infty}(1-2^{-j})} \leq \prod_{j=101}^{\infty} e^{2\sum_{j=101}^{\infty} 2^{-j}}\\
\leq e^{2^{-99}} \;.$$
Also, it is not hard to verify that $e^{2^{-99}} \cdot 3.47 < 3.5$.
\end{proof}
Thus, 
 $$t = \left\lceil {2\ln(1/\delta)}{M^2} \right\rceil  \leq 17\log(1/\delta) \;,$$
 This completes the proof of \Cref{KPL1}.
%
\end{proof}

We now state and prove  \Cref{KPL2} (which is Lemma 2.4 in \cite{kaufman2012weight}).
\begin{customlem}{\ref{KPL2}}
Let $f : \F_2^n \rightarrow \F_2$ be a function such that $\bias{f} \geq \epsilon > 0$ and let $\delta>0$. Then there exists directions $y_1,\ldots,y_t \in \F_2^m$ such that,
$$\prob{x}{f(x) = \text{Maj}\left(\derivative{\sum_{i \in I}y_i}{f}(x) : \emptyset \neq I \subseteq [t]\right)} \geq 1 - \delta \;,$$
where $t = \lceil 2\log(1/\epsilon) + \log(1/\delta) + 1 \rceil$.
\end{customlem}

\begin{proof}

The starting point is again \Cref{prop:bias}. As before we can sample $y_1,\ldots,y_t$ independently and using Hoeffding's inequality argue that  
$$\text{Sgn}\left(\frac{1}{t}\sum_{i=1}^{t}(-1)^{\derivative{y_i}{f}(x)}\right)$$ 
approximates $f$. This however does not yield the dependence we are looking for and so we shall use Chebyshev's inequality.
The main observation is that from the derivatives $\set{\derivative{y_i}{f}(x)}$ we can computed all the derivatives in $\set{\derivative{y}{f} : y \in \text{span}\set{y_1,\ldots,y_t}(x)}$, and that this set is pairwise independent. Let  
$$S(x,y_1,\ldots,y_t) = \sum_{\emptyset \neq I \subseteq [t]} (-1)^{f(x) + \derivative{\sum_{i \in I} y_i}{f}(x)} \;.$$
It is not hard to see that for $x, y_1,\ldots,y_t \in \F_2^m$ sampled uniformly at random the above summands are pairwise independent and that  $S \geq 0$ iff $f(x) = \text{Maj}\set{\derivative{\sum_{i \in I} y_i}{f}(x)}$. Simple application of Chebyshev's inequality yields, 
\begin{align*}
\mathrm{Pr}[f(x) \neq \text{Maj}\set{\derivative{\sum_{i \in I} y_i}{f}(x) : \emptyset \neq I \subseteq [t]}] &= \mathrm{Pr}[S < 0]\\
&\leq \mathrm{Pr}|S - \E[S]| > \E[S]]\\
&< \frac{\text{Var}(S)}{(2^{t}-1)^2\bias{f}^2} \\
&\leq \frac{1}{(2^{t}-1)\bias{f}^2} \;. 
\end{align*}
Thus, for $t = \log \frac{1}{\delta} + 2\log \frac{1}{\epsilon} + 1$ we get that the above probability is at most $\delta$.

\end{proof}

\section{Missing calculations from the proof of \Cref{main thm - low bias estimation}}
\label{section: small calculation}
In the proof of \Cref{main thm - low bias estimation} we claimed that the smallest natural number $s$ for which 
$$17 (2s+4)\gamma^{s-2} \leq \frac{1}{2}\left(1-\gamma\left(1+\frac{2\ell + s + 1}{m-(2\ell + s + 1)}\right)\right)^{2\ell+s+1} $$
satisfies $s =s(\gamma,\ell)= O\per{\frac{\gamma\ell + \log(1/1-2\gamma)}{1-2\gamma}+1}$. We now prove this upper bound on $s$.

We first deal with the case in which $\gamma$ is very close to $1/2$ as this is the most interesting case. 

\begin{lem}\label{lem:s'}
Let $m,\ell \in \N$ and $\rho < 1/4$ be a positive parameter, which may be constant or function of $m$.
Let $\gamma = 1/2-\rho$ and let $\tilde{\gamma}$ be such that $\gamma\leq \tilde{\gamma} \leq \frac{1}{2}-\frac{\rho}{2}$. Then, for $s = O\per{\frac{\gamma\ell + \log(1/1-2\gamma)}{1-2\gamma}+1}=O\per{\frac{\ell + \log\frac{1}{\rho}}{\rho}}$ we have,
$$17(2s+4)\gamma^{s-2} \leq \frac{1}{2}(1-\Tilde{\gamma})^{2\ell+s+1} \;.$$
\end{lem}
\begin{proof}
Define,
$$s' = \frac{8+\log\per{\frac{1}{1-\Tilde{\gamma}}}(2\ell+1)+2\log\frac{1}{\gamma}}{\log \frac{1-\Tilde{\gamma}}{\gamma}}\quad \text{ and } \quad \Delta =\frac{\log s' + \log\per{1+\frac{2}{\log\frac{1-\Tilde{\gamma}}{\gamma}}}}{\log \frac{1-\Tilde{\gamma}}{\gamma}};.
$$
and set $s = s'+\Delta$.

As $s \geq 1$ it suffices to prove that,
\begin{align}
\label{eq: need to prove}
2\cdot 17\cdot 6s\gamma^{s-2} = 204s\gamma^{s-2} \leq (1-\Tilde{\gamma})^{2\ell+s+1}\;.
\end{align}
We claim that since,
\begin{align}
\label{eq: inequality 1}
s' \geq \frac{8+\log\per{\frac{1}{1-\Tilde{\gamma}}}(2\ell+1)+2\log\frac{1}{\gamma}}{\log \frac{1-\Tilde{\gamma}}{\gamma}}    
\end{align}

we have,
$$204\gamma^{s'-2} \leq (1-\Tilde{\gamma})^{2\ell+s'+1}\;.$$
To see this re-write the inequality as,
$$204\gamma^{-2}(1-\Tilde{\gamma})^{-1}(1-\Tilde{\gamma})^{-2\ell} \leq \per{\frac{1-\Tilde{\gamma}}{\gamma}}^{s'}\;.$$
By taking the logarithm on both sides the above is equivalent to,
$$\log(204) + 2\log\frac{1}{\gamma} + \log \frac{1}{1-\Tilde{\gamma}}+2\ell \log\frac{1}{1-\Tilde{\gamma}} \leq s'\log \frac{1-\Tilde{\gamma}}{\gamma}\;.$$
Since $\log(204)< 8$ this is indeed a consequence of \cref{eq: inequality 1}. Recall that $s= \Delta + s'$ then by cancelling this inequality from \cref{eq: need to prove} we get that it suffices to prove,
$$(s'+\Delta) \leq \per{\frac{1-\Tilde{\gamma}}{\gamma}}^{\Delta}\;.$$
Substituting $\Delta$ we get that the above is equivalent to,
$$s' + \frac{\log s'}{\log\frac{1-\Tilde{\gamma}}{\gamma}} + \frac{\log\per{1+\frac{2}{\log\frac{1-\Tilde{\gamma}}{\gamma}}}}{\log\frac{1-\Tilde{\gamma}}{\gamma}} \leq s'\cdot \per{1 + \frac{2}{\log\frac{1-\Tilde{\gamma}}{\gamma}}}\;.$$
This inequality is equivalent to,
$$\log s'+\log\per{1+\frac{2}{\log\frac{1-\Tilde{\gamma}}{\gamma}}} \leq 2s'\;.$$
Clearly $\log s' \leq s'$ and $\log\per{1+\frac{2}{\log\frac{1-\Tilde{\gamma}}{\gamma}}} \leq s'$.

To complete the proof we thus have to show that $s =O\per{\frac{\gamma\ell + \log(1/1-2\gamma)}{1-2\gamma}}= O\per{\frac{\ell + \log\frac{1}{\rho}}{\rho}}$. I.e. that
$$  \frac{8+\log\per{\frac{1}{1-\Tilde{\gamma}}}(2\ell+1)+2\log\frac{1}{\gamma}}{\log \frac{1-\Tilde{\gamma}}{\gamma}}+ \frac{\log s' + \log\per{1+\frac{2}{\log\frac{1-\Tilde{\gamma}}{\gamma}}}}{\log \frac{1-\Tilde{\gamma}}{\gamma}}= O\per{\frac{\ell + \log\frac{1}{\rho}}{\rho}}\;.$$
As $\gamma,\tilde{\gamma}$ are bounded away from $1$ and $0$ we have that $\log\frac{1}{\gamma}, \log\per{\frac{1}{1-\Tilde{\gamma}}} = O(1)$. In addition,
$$\log \frac{1-\Tilde{\gamma}}{\gamma} \geq  \log \frac{\frac{1+\rho}{2}}{\frac{1}{2}-\rho}=\log\frac{1+\rho}{1-2\rho}=\log\per{1+\frac{3\rho}{1-2\rho}} \geq \frac{\rho}{1-2\rho}\;. $$
We thus see that $s' = O\per{\ell/\rho}$ and that $\Delta = O\per{\frac{\log s' + \log(O(1/\rho))}{\rho}} $. Hence, $s = O\per{\frac{\ell + \log\frac{1}{\rho}}{\rho}}$.
\end{proof}

Recall that in the case of \Cref{main thm - low bias estimation} $\tilde{\gamma} = \gamma\per{1 + \frac{t}{m-t}}$ where $t=2\ell+s+1$ and $\gamma = \frac{1}{2}-\rho$. A simple calculation reveals that if  $t \leq \frac{\rho}{1-\rho}m$ then $\tilde{\gamma} \leq \frac{1}{2}-\frac{\rho}{2}$. Thus, for the assumption in \Cref{lem:s'} to hold it suffices that
$$\frac{\ell + \log\frac{1}{\rho}}{\rho^2} = o(m)\;.$$

We now consider the case $\gamma\leq 1/4$. As before it is enough to prove that
$$204 s \gamma^{-2}(1-\Tilde{\gamma})^{-1}(1-\Tilde{\gamma})^{-2\ell} \leq \per{\frac{1-\Tilde{\gamma}}{\gamma}}^{s}\;.$$
Let us assume that $\tilde{\gamma}\leq 2\gamma\leq 1/2$. We thus have that $1-\tilde{\gamma} \geq 2\gamma, \sqrt{\gamma}$ and $\frac{1}{1-\tilde{\gamma}}\leq \exp(2\tilde{\gamma})$.
We thus get
$$204 s \gamma^{-2}(1-\Tilde{\gamma})^{-1}(1-\Tilde{\gamma})^{-2\ell} \leq 204 s \per{\frac{1-\Tilde{\gamma}}{\gamma}}^{4}
\exp\per{6\tilde{\gamma}\ell}\leq 204 s \per{\frac{1-\Tilde{\gamma}}{\gamma}}^{4}
\exp\per{12 {\gamma}\ell}\;.$$
In addition for $s \geq {24\gamma \ell}+256$ we have
$$\per{\frac{1-\Tilde{\gamma}}{\gamma}}^{s} \geq \per{\frac{1-\Tilde{\gamma}}{\gamma}}^{4} \exp_2\per{24\gamma\ell+252}\;.$$
It therefore suffices to prove that $204 s \leq \exp_2\per{12\gamma\ell+252} $. This follows as for all $s\geq 256$ we have that
$256s \leq \exp_2\per{s/2}$, and for our choice of $s$ it holds that $s/2 < 12\gamma\ell+252$. All that is left to prove is that when $\gamma \leq 1/4$ for some $s=O\per{\frac{\gamma\ell + \log(1/1-2\gamma)}{1-2\gamma}+1}$ it holds that $s \geq {24\gamma \ell} +256$. This clearly holds as $\frac{1}{1-2\gamma} >1$.

Recall that we assumed that $\tilde{\gamma}\leq 2\gamma$. For this to hold it suffices to require in \Cref{main thm - low bias estimation} that  $\ell= o(m)$.
\end{document}